\documentclass[journal]{IEEEtran}
\usepackage[utf8]{inputenc}
\usepackage[T1]{fontenc}
\usepackage{url}
\usepackage{ifthen}
\usepackage{cite}
\usepackage[cmex10]{amsmath}
\usepackage{stfloats}
\usepackage{graphicx}
\usepackage{epstopdf}
\usepackage{amsfonts}
\usepackage{enumitem}
\usepackage{amssymb}
\usepackage{mathrsfs}
\usepackage{bm}
\usepackage{mdwmath}
\usepackage{mdwtab}
\usepackage{dsfont}
\usepackage{color}
\usepackage{verbatim}
\usepackage{array}
\usepackage{epstopdf}
\usepackage{multirow}
\usepackage{booktabs}
\usepackage{lineno}
\usepackage{amsmath}
\usepackage{makecell}
\usepackage{algorithmic}

\usepackage{amsthm}
\usepackage{tabularx}
\usepackage{graphicx}
\usepackage{wrapfig}
\usepackage{cite}
\usepackage{subcaption}

\usepackage{flushend}
\usepackage[ruled,vlined,linesnumbered]{algorithm2e}

\allowdisplaybreaks

\newtheorem{Prob}{Problem}

\newtheorem{theorem}{Theorem}

\usepackage{titletoc}

\renewcommand\thesubsubsection{\alph{subsubsection}}

\begin{document}
\title{Energy Optimization of Multi-task DNN Inference in MEC-assisted XR Devices: A Lyapunov-Guided Reinforcement Learning Approach}


\author{Yanzan Sun,~\IEEEmembership{Member,~IEEE}, Jiacheng Qiu, Guangjin Pan, Shugong Xu,~\IEEEmembership{Fellow,~IEEE}, \\
    Shunqing Zhang,~\IEEEmembership{Senior Member,~IEEE}, Xiaoyun Wang, and Shuangfeng Han,~\IEEEmembership{Senior Member,~IEEE}
	\thanks{
		  Y. Sun, J. Qiu, G. Pan, S. Xu, and S. Zhang are with the School of Communication and Information Engineering, Shanghai University, Shanghai 200444, China. Emails: \{yanzansun, qjccc, guangjin\_pan, shugong, shunqing\}@shu.edu.cn. (Corresponding author: Guangjin Pan.)

           X. Wang and S. Han are with China Mobile Research Institute, Beijing 100053, China. Emails: \{wangxiaoyun, hanshuangfeng\}@chinamobile.com

	This work was supported by the National Key Research and Development Program of China under Grant 2022YFB2902005, 2022YFB2902304, and 2022YFB2902002, the National Natural Science Foundation of China (NSFC) under Grants 62071284. }
}
    \maketitle

\markboth{Journal of \LaTeX\ Class Files,~Vol.~XX, No.~XX, January~2024}%
{Shell \MakeLowercase{\textit{et al.}}: Bare Demo of IEEEtran.cls for IEEE Transactions on Magnetics Journals}
%



\IEEEtitleabstractindextext{%
\begin{abstract}

Extended Reality (XR), blending virtual and real worlds, is a key application of future networks. While AI advancements enhance XR capabilities, they also impose significant computational and energy challenges on lightweight XR devices. In this paper, we developed a distributed queue model for multi-task DNN inference, addressing issues of resource competition and queue coupling. In response to the challenges posed by the high energy consumption and limited resources of XR devices, we designed a dual time-scale joint optimization strategy for model partitioning and resource allocation, formulated as a bi-level optimization problem. This strategy aims to minimize the total energy consumption of XR devices while ensuring queue stability and adhering to computational and communication resource constraints. To tackle this problem, we devised a Lyapunov-guided Proximal Policy Optimization algorithm, named LyaPPO. 
Numerical results demonstrate that the LyaPPO algorithm outperforms the baselines, achieving energy conservation of 24.79\% to 46.14\% under varying resource capacities. Specifically, the proposed algorithm reduces energy consumption of XR devices by 24.29\% to 56.62\% compared to baselines algorithms.
\end{abstract}

}

\maketitle

\IEEEdisplaynontitleabstractindextext

%
\IEEEpeerreviewmaketitle

\section{Introduction}
\IEEEPARstart{T}{he} advent of the Metaverse \cite{t5} has ignited considerable interest in immersive experiences within virtual environments. A cornerstone technology enabling these experiences is Extended Reality (XR), which has rapidly gained prominence as a key 5G media application \cite{t1,XR_1,t3,XR_2}. XR applications strive to elevate user interaction by continuously analyzing user behaviors and environmental contexts through tasks such as gesture recognition \cite{gesture1}, speech recognition \cite{speech1}, and object tracking \cite{ObjectTrackin}. The use of deep neural network (DNN)-based AI algorithms, known for their high precision, has further enhanced these applications, significantly improving the quality of user experiences.

Despite the substantial potential of these AI-driven solutions in XR, they impose notable challenges for XR devices. On one hand, DNN-based algorithms demand intensive computational resources, resulting in high latency and increased energy consumption \cite{intro1, 14}. On the other hand, XR devices, due to their lightweight design, are constrained by limited battery capacity and computational capabilities, restricting the direct implementation of DNN models \cite{t3,t7,t6}. Moreover, the intricate interaction requirements of XR applications often necessitate the simultaneous operation of multiple DNN models, compounding these computational and energy demands.

Mobile Edge Computing (MEC) has emerged as a promising solution to address these limitations. By offloading computationally intensive tasks to edge servers, edge AI offers a viable approach to reduce network load, decrease latency, and lower energy consumption on devices \cite{intro3}. Specifically, DNN model partitioning enables the division of model processing between XR devices and MEC servers \cite{03, 20}. This hybrid processing approach capitalizes on the computational resources of both the device and the network edge, thus enhancing overall system performance.

In response to these considerations, this study aims to optimize the inference of multi-task AI models in XR applications supported by MEC. We propose a Lyapunov-guided reinforcement learning (DRL) approach to minimize energy consumption for DNN inference in MEC-assisted XR systems. The contributions of this paper are as follows.

\begin{itemize}
\item{\textbf{Multi-Task DNN Inference for MEC-assisted XR Applications.}}
In this paper, we consider an edge network architecture for MEC-assisted XR applications focused on multi-task DNN inference. Within this architecture, each XR device is required to handle multiple AI applications simultaneously, corresponding to multi-task DNN inference. For each application, there exists a coupled relationship among the local queues, transmission queues, and MEC queues. These queues are coupled because they share dependencies and participate in the competition for communication and computation resources. We also consider utilizing DNN partitioning techniques to fully leverage the computational capabilities of both XR devices and MEC. Based on this, our aim is to optimize the allocation of system resources and the DNN partition point for each AI application, reducing the energy consumption of XR devices while ensuring the completion of DNN inference tasks.

\item{\textbf{Bi-level Modeling for Dual Time-Scale Optimization.}} 
Since DNN partitioning cannot be reconfigured frequently and requires adjustments on a larger time scale, while resource allocation often occurs on a smaller time scale to ensure system flexibility, this introduces a dual time-scale challenge for system optimization. Therefore, we model the problem as a bi-level optimization problem. In the upper-level optimization, we adjust the DNN partition points in each partition adjustment period. In the lower-level optimization, we jointly optimize the allocation of communication and computation resources in each time slot to ensure queue stability while reducing system energy consumption.

\item{\textbf{Lyapunov-Guided DRL Solution.}}
To solve the proposed bi-level problem, we introduce a Lyapunov-guided Proximal Policy Optimization (LyaPPO) algorithm. For the lower-level optimization problem, based on the Lyapunov optimization method, we first reformulate the problem as a single-slot problem and further decompose it into three sub-problems, i.e., the local computing subproblem, the transmission subproblem, and the edge computing subproblem. For the upper-level optimization problem, we employ the Proximal Policy Optimization (PPO) algorithm to determine the DNN partition points.
\end{itemize}

Simulation results demonstrate that the proposed LyaPPO algorithm reduces XR device energy consumption by 46.14\%, 29.10\%, and 24.79\% under varying local computational resources, maximum transmit power, and edge computational resources, respectively, compared to baseline methods. Our approach dynamically adjusts between local computing and transmission based on energy costs, optimizing performance by balancing resource allocation between the device and MEC.

The remainder of this paper is organized as follows. Section \ref{Sec2} provides an in-depth review of related work. Section \ref{Sec3} outlines the system models, including DNN inference, computation and communication frameworks, and the queue model, followed by the formulation of a dual time-scale bi-level energy optimization problem. Section \ref{Sec4} proposes the LyaPPO algorithm to solve this problem. Numerical simulations and analysis are presented in Section \ref{Sec5}, and the paper concludes in Section \ref{Sec6}.

\section{Related Works} \label{Sec2}
A substantial body of research leverages the computational capacity of MEC to enable edge-assisted collaborative inference. Effective techniques in this domain include computation offloading, model partitioning, etc., which have proven their efficacy in various applications. Specifically, in the field of computation offloading, Fang et al. \cite{15} design TORA-DRL algorithm, which employ DRL to optimize power consumption and alleviate network load. HybridPPO in \cite{19} applies task offloading to reduce latency and energy consumption under specific server constraints.
Bi et al. \cite{01} present the LyDROO framework, which optimally manages data processing through resource allocation and offloading. Wu et al. \cite{07} address stochastic offloading with perturbed Lyapunov optimization to enhance energy efficiency.
Similarly, the work in \cite{33} also uses Lyapunov optimization to minimize energy consumption, and designs an energy efficient dynamic offloading algorithm.
Dai et al. \cite{34} integrate digital twin technology with IIoT networks to model network topology and random task arrivals, proposing an asynchronous actor-critic algorithm to optimize long-term energy efficiency.

Despite the benefits of computation offloading, fully transferring the entire neural network inference tasks to MEC servers often limits the efficient utilization of both device and MEC server resources. Additionally,  the transmission of massive data from devices to MEC server for processing generates large amounts of cross-network traffic, consuming more network resources and increasing energy consumption \cite{29}. Computation offloading faces challenges under dynamic communication conditions, which can impact connectivity and reliability. Therefore, this study identifies DNN model partitioning as a promising alternative for edge collaborative inference. The model partitioning technology strategically segments the DNN models into multiple parts in accordance with its multilayered structure \cite{22}, which allows partial inference tasks to be processed locally before offloading to MEC servers.

Effective queue management is crucial in model partitioning to ensure system efficiency. Several studies incorporate queue mechanisms within model partitioning techniques to address dynamic resource demands. For instance, an M/D/1 queuing model in \cite{04} aims to minimize end-to-end (E2E) latency by jointly optimizing partitioning and resource allocation. COSREL, proposed in \cite{24}, is an online DRL co-scheduling framework that utilizes heterogeneous computing resources (CPUs, GPUs,  DSPs) for concurrent inference of multiple DNN models to enhance throughput, reduce latency, and improve energy efficiency. Ale et al. \cite{31} propose the Dirichlet Deep Deterministic Policy Gradient (D3PG) framework to jointly optimize task partitioning, computational offloading, and computational frequency control, where subtasks are processed sequentially in the queues. 

Lyapunov optimization technique is adopted to further control the stability of the queues while dynamically allocating resources \cite{03, 18, 23, 25}. The work in \cite{03} designs a queue system comprising DNN task load queues and energy queues, enabling online control of task partitioning and offloading to jointly optimize latency and energy consumption. Jiang et al. \cite{18} build an online joint offloading and resource allocation framework (JORA) by employing Lyapunov optimization to create virtual energy queues for reducing latency and energy consumption. Su et al. \cite{23} present the DDPRA algorithm, which combines Lyapunov optimization with reinforcement learning to dynamically partition DNNs and allocate resources, minimizing the long-term average delay under energy constraints. RT-DMP described in \cite{25} optimally balances energy consumption, throughput, and E2E latency through joint DNN partitioning and resource adaptation, addressing the quality of experience (QoE) demands of mobile environments with a virtual queue-based Lyapunov framework. Furthermore, an upsurge interest in energy efficiency can be observed in the literature e.g., \cite{03, 24, 18, 23, 25}. Unfortunately, these studies often overlook the distributed structure of the queues between the device and server sides in the model partitioning scenario, whereas distributed queues require optimization by considering the coupling between them.

Complex scenarios in edge collaborative inference also require support for diverse AI applications across heterogeneous devices. For example, Gao et al. \cite{11} tackle DNN partitioning and resource allocation under multi-user constraints, reducing worst-case latency for real-time applications. 
Yang et al. \cite{32} propose a MEC-based hierarchical machine learning (ML) task distribution framework that aims to minimize total delay, considering constraints such as model complexity, inference error rate, data quality, and available resources. However, studies \cite{03, 04, 23, 11, 32} address scenarios with heterogeneous devices but do not consider the use of multiple DNN models. In contrast, the work in \cite{23} support heterogeneous devices and deploy a different DNN model on each device, enabling a variety of AI services across the scenario. Regrettably, neither approach considers scenarios where multiple DNN models are deployed on a single heterogeneous device.

Additionally, the above references generally do not address the dual time-scale optimization of model partitioning and resource allocation. This optimization approach can significantly reduce the configuration overhead of DNN models when adjusting partition points, thereby improving system efficiency in dynamic environments.

\section{System Model and problem formulation} \label{Sec3}

\begin{figure*}[tb]
\centering 
\includegraphics[scale=0.4]{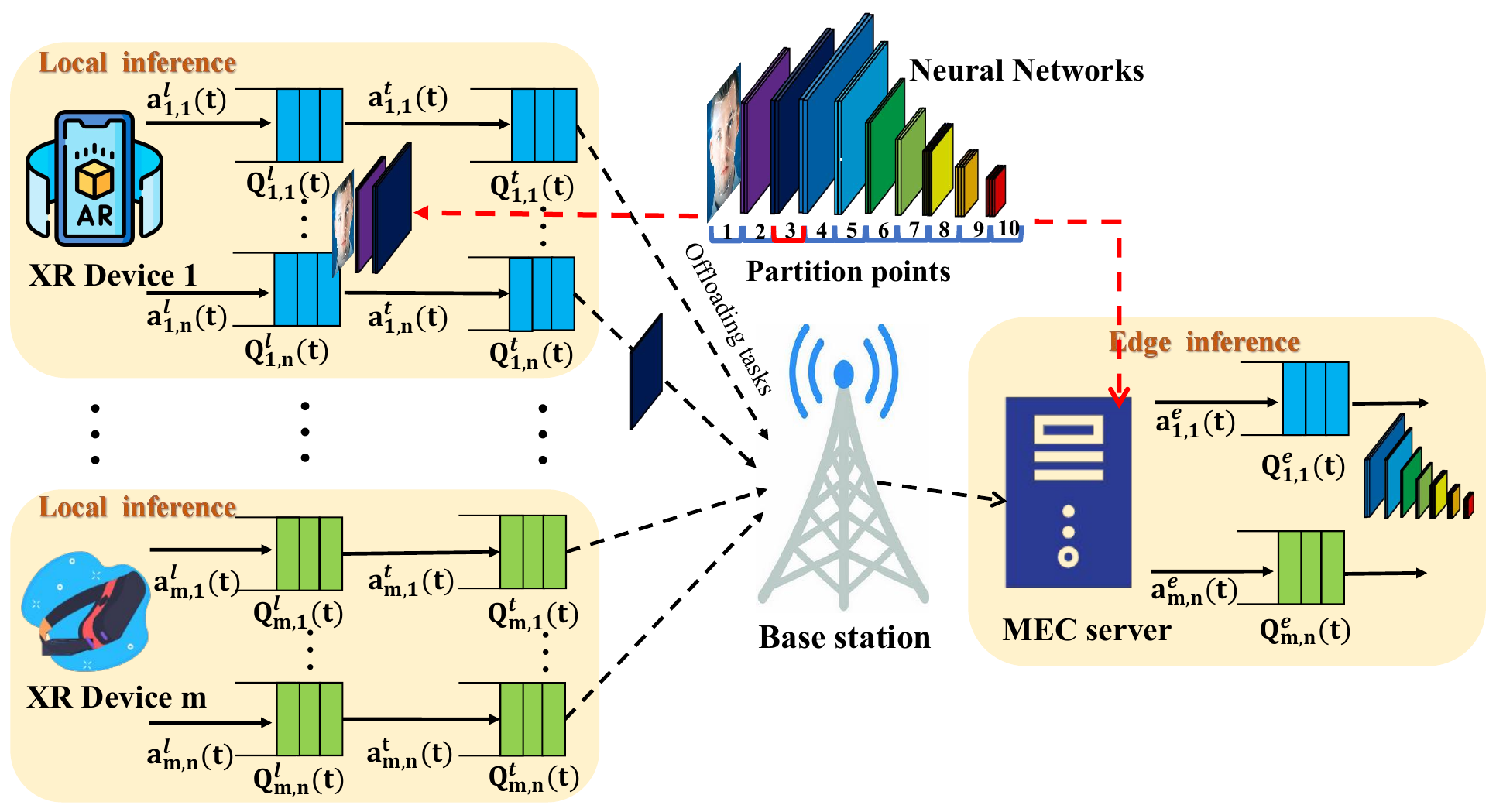}
\caption{MEC-assisted Collaborative Inference Architecture for AI Applications on XR Devices.}
\label{fig-system} 
\end{figure*}

As illustrated in Fig. \ref{fig-system}, we consider a MEC system comprising a single base station (BS) and $M$ XR devices, denoted by the set $\mathcal{M} = \{1, 2, \ldots, M\}$. Each XR device $m \in \mathcal{M}$ hosts various AI services, such as facial recognition \cite{facial, facial2} and gesture recognition \cite{gesture1}. To support these services, each XR device $m$ deploys $N_m$ DNN models, represented by $\mathcal{N}_m = \{1, 2, \ldots, N_m\}$, where each model corresponds to a specific AI task required for the XR application. 

Due to the lightweight design constraints of XR devices, they typically operate under strict energy limitations. To mitigate energy consumption, we consider that XR devices can offload portions of their DNN inference tasks to an MEC server through a process called DNN partitioning \cite{03,20}. Specifically, for each service (i.e., each DNN model) on an XR device, initial inference computation is performed locally. The intermediate results are then transmitted over a wireless network to the MEC server, where the remaining computation is completed. This collaborative approach fully utilizes the computational resources of XR devices and MEC while simultaneously reducing the energy consumption of XR devices. Therefore, for each AI service, DNN inference tasks are generated randomly. Throughout the lifecycle of a task, from initiation to computation completion, tasks traverse through three distinct distributed queues: the local queue, the transmission queue, and the edge queue. The entire execution of each task is divided into the following three steps,
\begin{itemize}
    \item{\textit{Step 1.}} DNN inference tasks are generated and queued in the local queues. In this phase, computation before the model partition point is handled locally by the XR device.

    \item{\textit{Step 2.}} After local inference, the remaining computing task is offloaded to the MEC server. This step involves transmitting intermediate layer features of the task into the transmission queues.
   
    \item{\textit{Step 3.}} The MEC server receives the task in the edge queues, where it processes the remaining computations after the model partition point.
\end{itemize}

To model the temporal dynamics of the AI inference system for XR applications, we segment the time domain into a set of discrete time slots, represented by $\mathcal{T} = \{0, 1, 2, \ldots, T_s\}$, each with a duration of $\tau$. Computational and communication resources are adjustable in each time slot $t_s$. However, frequent adjustments to the model's partitioning decision are impractical due to the significant overhead involved in loading and reconfiguring the model. Therefore, we make adjustments to the partition decisions of the DNN model on a larger timescale. Specifically, we denote $G \tau$ as the duration of each partition adjustment period, and define the partition adjustment intervals within the set $\mathcal{T}_p =\{0,G,2G,\cdots,GT_p\} \subset \mathcal{T}$. For convenience, we define $t_p$ as the partition adjustment period in time slot $t_s G$. A summary of the key notations and definitions used in this model is provided in Table \ref{table symbol summmary}.

\begin{table}[ht]
  \centering
  \caption{SUMMARY OF MAIN NOTATIONS}\label{table symbol summmary}
  \begin{tabularx}{\linewidth}{lX}
    \toprule
    \textbf{Notation} & \textbf{Description} \\
    \midrule
    $M$ & Number of XR devices\\
    $N_m$ & Number of deployed models on device $m$\\
    $m$ & Device index\\
    $n$ & Service index\\
    $\tau$ & Time slot duration\\
    $G$ & Partition points update interval\\
    $t_s$ & Time slot\\
    $t_p$ & Partition adjustment period\\
    $k_{m,n}^{t_p}$ &  Partition points of service $n$ on XR device $m$ in partition adjustment period $t_p$\\
    $c_{m,n}^{local}(t_p)$ & The computational complexity of the inference task on the XR device\\
    $c_{m,n}^{edge}(t_p)$ & The computational complexity of the inference task on the MEC server\\
    $d_{m,n}(t_p)$ & The size of the output feature map transferred from the XR device to the MEC server\\
    $a_{m,n}(t_s)$ & Number of new task arrivals in queues in time slot $t_s$\\
    $a_{m,n}^l(t_s)$ & Number of tasks backlogged in local queues in time slot $t_s$\\
    $a_{m,n}^t(t_s)$ & Number of tasks backlogged in transmission queues in time slot $t_s$\\
    $a_{m,n}^e(t_s)$ & Number of tasks backlogged in edge queues in time slot $t_s$\\
    $F_m^l$ & Maximum local computational resource of XR device $m$\\
    $F^e$ & Maximum edge computational resource of the MEC server\\
    $R_m(t_s)$ & Maximum transmission rate of XR device $m$ in time slot $t_s$\\
    $f_{m,n}^l(t_s)$ & Allocated computational frequency of each local queue  \\
    $r_{m,n}(t_s)$ & Allocated transmission rate of each transmission queue\\
    $p_m(t_s)$ & Allocated transmit power of XR device $m$\\
    $f_{m,n}^e(t_s)$ & Allocated computational frequency of each edge queue\\
    $Q_{m,n}^l(t_s)$ & Workload of local queue for XR device $m$ service $n$ in time slot $t_s$\\
    $Q_{m,n}^t(t_s)$ & Workload of transmission queue for XR device $m$ service $n$ in time slot $t_s$\\
    $Q_{m,n}^e(t_s)$ & Workload of edge queue for XR device $m$ service $n$ in time slot $t_s$\\   
    $h_m(t_s)$ & Channel gain of XR device $m$\\
    \bottomrule
  \end{tabularx}
\end{table}

\subsection{DNN Inference Task Model}
In the AI collaborative inference scenario tailored for XR applications, the AI service requirements of XR devices are addressed through DNN models. We assume that each DNN model for an AI service has a fixed model architecture and a constant input data size. To enable uniform evaluation of the computational complexity associated with DNN inference tasks deployed on XR devices, we utilize the number of multiply-and-accumulate operations (MACs) as a standard metric \cite{05}.

For the service $n$ on device $m$, the total MACs required for completing a single DNN inference task is represented by $C_{m,n}$ (in MACs). The input data size is denoted by $D_{m,n}$ (in bits), while the maximum number of partition layers in the model is represented by $K_{m,n}$. The partition point within the partition adjustment period $t_p$ is denoted by $k_{m,n}^{t_p} \in \mathcal{K}_{m,n} \triangleq \{1, 2, \dots, K_{m,n}\}$.

We denote the proportion of computational complexity required by the layers preceding the partition point $k_{m,n}^{t_p}$ as $c_{m,n}^{k_{m,n}^{t_p}}$, with the remaining proportion expressed as $1 - c_{m,n}^{k_{m,n}^{t_p}}$. These proportions satisfy the ordering $c_{m,n}^0 < c_{m,n}^1 < \dots < c_{m,n}^{K_{m,n}}$, where $c_{m,n}^0 = 0$ signifies complete offloading of the task to the edge server, and $c_{m,n}^{K_{m,n}} = 1$ implies full local execution on the XR device.

Similarly, $d_{m,n}^{k_{m,n}^{t_p}}$ quantifies the ratio of the output feature map size at partition point $k_{m,n}^{t_p}$ to the input data size $D_{m,n}$. Specifically, $d_{m,n}^{K_{m,n}} = 0$ indicates computation of the entire inference task locally on the XR device, while $d_{m,n}^0 = 1$ indicates offloading of the entire inference task to the MEC server. Note that $d_{m,n}^{k_{m,n}^{t_p}}$ can exceed 1, as intermediate layer outputs may surpass the size of the neural network's input data.

Based on these definitions, a comprehensive DNN inference task model can be established to systematically evaluate task workloads from initiation to completion. The computational complexity of the inference task on the XR device is given by ${c}^{local}_{m,n}(t_p) = c_{m,n}^{k_{m,n}^{t_p}}C_{m,n}\rho$, where $\rho$ (in cycles/MAC) is the CPU cycles required for each multiplication and addition operation, which is dependent on the CPU model \cite{05}. The size of the output feature map transferred from the XR device to the MEC server is expressed as ${d}_{m,n}(t_p) = d_{m,n}^{k_{m,n}^{t_p}}D_{m,n}$. Meanwhile, the computational complexity remaining for the MEC server is given by ${c}^{edge}_{m,n}(t_p) = (1 - c_{m,n}^{k_{m,n}^{t_p}})C_{m,n}\rho$.

\subsection{Computation and Communication Models}
In the MEC-assisted XR system, distinct queues are established for managing AI service requests originating from each XR device $m$ and its associated model $n$. Due to the stochastic nature of wireless network conditions and the random arrival of tasks, dynamic allocation of computation and communication resources is crucial to optimizing overall system performance in real time.

The computational resources allocated to service $n$ by XR device $m$ and the MEC server are denoted by $f_{m,n}^l(t_s)$ and $f_{m,n}^e(t_s)$ (in cycles/s), respectively. These resource allocations are constrained by the available computational capacities $F_m^l$ at the XR device and $F^e$ at the MEC server respectively, which can be expressed as
\begin{eqnarray}\label{equ2-B-2}
&&\sum_{n \in \mathcal{N}_m} f_{m,n}^l(t_s) \le F_{m}^{l}, \ \forall m \in \mathcal{M},\\
&&\sum_{m \in \mathcal{M}}\sum_{n \in \mathcal{N}_m}f_{m,n}^e(t_s) \le F^{e}.
\end{eqnarray}

To simplify the problem, we assume that the system bandwidth $B_w$ is evenly divided into $M$ orthogonal channels to serve the $M$ XR devices. The maximum transmission rate for XR device $m$ in time slot $t_s$ can be expressed as
\begin{eqnarray}
R_m(t_s)=
b_w log_2 \left( 1+\frac{p_m(t_s)h_m(t_s)}{b_w N_0} \right)
\end{eqnarray}
 where $h_m(t_s)$ and $p_m(t_s)$ represent the channel gain and the allocated transmission power, respectively, while $b_w = \frac{B_w}{M}$ is the bandwidth allocated to each XR device. For AI service $n$ on device $m$, the allocated transmission rate $r_{m,n}(t_s)$ must satisfy the constraint
\begin{eqnarray}\label{equ:constraint-rate}
\sum_{n \in \mathcal{N}_m}r_{m,n}(t_s) \le R_m(t_s).
\end{eqnarray}

Therefore, the energy consumption of XR device $m$ in time slot $t_s$ comprises computational energy $E^{l}(t_s)$ and transmission energy $E^{t}(t_s)$, expressed as
\begin{eqnarray}
&& E_m^{l}(t_s)= \tau\delta(\sum_{n \in \mathcal{N}_m} f_{m,n}^l(t_s))^3,\\
&& E_m^{t}(t_s)= \tau p_m(t_s).
\end{eqnarray}
where $\delta$ is the energy coefficient.

\subsection{Queue model}

The arrival process of inference tasks generated from the corresponding AI service follows a specific random distribution.  To effectively manage these tasks, distributed queue modeling is required, encompassing local queues, transmission queues, and edge queues. The queues $Q_{m,n}^l(t_s)$ and $Q_{m,n}^e(t_s)$ (in cycles) represent the remaining computational workload backlog on XR device and MEC server, respectively. The queue $Q_{m,n}^t(t_s)$ (in bits) represents the remaining transmission data backlog in the XR device. All queues adhere to a first-come, first-served scheduling policy. We define $[x]^+=\mathop{\textrm{max}}\{x, 0\}$. The updates for each of the queues are expressed as follows:
\begin{align}
Q_{m,n}^l(t_s+1)\!=\![Q_{m,n}^l(t_s)-f^{l}_{m,n}(t_s)\tau+a^l_{m,n}(t_s)c^{local}_{m,n}(t_p)]^+,\label{Ql_update}
\end{align}
\begin{align}
\!Q_{m,n}^t(t_s+1)\!=\![Q_{m,n}^t(t_s)-r_{m,n}(t_s)\tau+a^t_{m,n}(t_s)d_{m,n}(t_p)]^+,\label{Qt_update}
\end{align}
\begin{align}
\!Q_{m,n}^e(t_s+1)\!=\![Q_{m,n}^e(t_s)-f^{e}_{m,n}(t_s)\tau+a^e_{m,n}(t_s)c^{edge}_{m,n}(t_p)]^+.\label{Qe_update}
\end{align}
where $a^l_{m,n}(t_s)$, $a^t_{m,n}(t_s)$ and $a^e_{m,n}(t_s)$ indicate the number of DNN inference tasks arriving at the local queue, transmission queue, and edge queue, respectively, in time slot $t_s$, and can be calculated by
\begin{eqnarray}\label{al_update}
&\! \! \! \! \! \! \! \!a^l_{m,n}(t_s+1) \!=\! \left\{
\begin{aligned}
&[a^l_{m,n}(t_s)-b_{m,n}^l(t_s)]^+,\qquad\qquad \: \: k_{m,n}^{t_p}=0\\
&[a_{m,n}(t_s)+a^l_{m,n}(t_s)-b_{m,n}^l(t_s)]^+, \: \: \mathop{\textrm{others}}\\
\end{aligned} 
\right.
\end{eqnarray}
\begin{eqnarray}\label{at_update}
&\! \! \! \! \! \! \! \!a^t_{m,n}(t_s+1) \!=\! \left\{
\begin{aligned}
&[a_{m,n}(t_s)+a_{m,n}^t(t_s)-b_{m,n}^t(t_s)]^+,\: \: k_{m,n}^{t_p}=0\\
&[a_{m,n}^t(t_s)-b_{m,n}^t(t_s)]^+,\: \: k_{m,n}^{t_p}=K_{m,n}\\
&[b_{m,n}^l(t_s)+a_{m,n}^t(t_s)-b_{m,n}^t(t_s)]^+, \: \: \mathop{\textrm{others}}\\
\end{aligned} 
\right.
\end{eqnarray}
\begin{eqnarray}\label{ae_update}
&\! \! \! \! \! \! \! \!a^e_{m,n}(t_s+1)  \!=\! \left\{
\begin{aligned}
&[a^e_{m,n}(t_s)-b^e_{m,n}(t_s)]^+,\: \:k_{m,n}^{t_p}=K_{m,n}\\
&[b^t_{m,n}(t_s)+a^e_{m,n}(t_s)-b^e_{m,n}(t_s)]^+,\mathop{\textrm{others}}
\end{aligned} 
\right.
\end{eqnarray}
where $a_{m,n}(t_s)$ represents the new inference tasks generated in time slot $t_s$. The variable $b_{m,n}^l(t_s)$, $b_{m,n}^t(t_s)$, and $b_{m,n}^e(t_s)$ indicate the number of processed DNN inference tasks removed from $Q_{m,n}^l(t_s)$, $Q_{m,n}^t(t_s)$ and $Q_{m,n}^e(t_s)$, respectively.

\subsection{Problem Formulation}
The primary objective of this paper is to minimize the long-term average energy consumption across all XR devices, while ensuring queue stability and satisfying the constraints on computational and communication resources. Given that partition decisions and resource allocations occur on two different time scales, we model this problem as a bi-level optimization problem under a dual time-scale. Let $\bm{k}^{t_p}\triangleq\{k_{m,n}^{t_p},\forall n,\forall m\}$, $\bm{r}(t)\triangleq\{r_{m,n}(t_s),\forall n\}$, $\bm{p}(t)\triangleq\{p_{m}(t_s),\forall m\}$, $\bm{f}^l(t) \triangleq \{f_{m,n}^l(t_s), \forall n, \forall m\}$, $\bm{f}^e(t) \triangleq \{f_{m,n}^e(t_s), \forall n, \forall m\}$.

\begin{Prob}[Bi-level Optimization Problem] The partition decision and resource allocation tasks of DNN based on dual time-scale can be formulated as a bi-level optimization problem.
\begin{align}
        \underset{\bm{k}^{t_p}}{\textrm{min}} &  \!    \lim\limits_{T_p\to+\infty} \sum_{t_p=1}^{T_p}\frac{1}{T_p}\left[\sum_{t_s=t_pG}^{t_pG+G-1}E^*(t_s,\bm{k}^{t_p})\right] \nonumber\\
        \textrm{s.t.} \nonumber \\ & \! \! \! E^*(t_s,\bm{k}^{t_p}) \triangleq 
        \begin{cases}
            &  \! \! \! \! \underset{\substack{\bm{r}(t), \bm{f}^l(t), \\ \bm{p}(t), \bm{f}^e(t)}} {\textrm{min}} \! \!  \!  
            \lim\limits_{T_s\to+\infty} \sum_{t_s=1}^{T_s}\frac{1}{T_s} [\displaystyle\sum_{{m \in \mathcal{M}}}  \big(E_{m}^l(t_s) \nonumber \\
            & \! \qquad \qquad \qquad \qquad \qquad + E_{m}^t(t_s)\big)] \nonumber
            \vspace{-0.1cm} \\
            & \textrm{s.t.} \nonumber \\
            &  C.1: \! \lim\limits_{T\to+\infty} \frac{1}{T} \displaystyle\sum_{t_s =0}^{T} \mathbb{E}[Q^l_{m,n}(t_s)]<\infty,  \\  
            &  C.2: \! \lim\limits_{T\to+\infty} \frac{1}{T} \displaystyle\sum_{t_s =0}^{T} \mathbb{E}[Q^t_{m,n}(t_s)]<\infty,   \\
            & C.3: \!  \lim\limits_{T\to+\infty} \frac{1}{T} \displaystyle\sum_{t_s =0}^{T} \mathbb{E}[Q^e_{m,n}(t_s)]<\infty,   \\
            & C.4: \displaystyle\sum_{n \in \mathcal{N}_m} f_{m,n}^l(t_s) \le F_{m}^{l}, \ \forall m \\
            & C.5: 0 \le p_m(t_s) \le p_m^{max}, \quad \forall{m} \\
            & C.6: \displaystyle\sum_{n \in \mathcal{N}_m}r_{m,n}(t_s) \le R_m(t_s), \quad  \forall m \\
            & C.7: \displaystyle\sum_{m \in \mathcal{M}}\sum_{n \in \mathcal{N}_m}f_{m,n}^e(t_s) \le F^{e}, \\
            & C.8: 0 \le f_{m,n}^l(t_s) \le \frac{Q_{m,n}^l(t_s)}{\tau}, \nonumber \\
            & C.9: 0 \le r_{m,n}(t_s) \le \frac{Q_{m,n}^t(t_s)}{\tau} ,\\
            & C.10: 0 \le f_{m,n}^e(t_s) \le \frac{Q_{m,n}^e(t_s)}{\tau} ,
            \end{cases} \\
        & C.11: k_{m,n}^{t_p} \in \{1, 2, \dots, K_{m,n}\}
        . \nonumber
\end{align} 
\vspace{-6mm}
\end{Prob}

Constraints $C.1$-$C.3$ are imposed to ensure the long-term stability of local, transmission, and edge queues, respectively. Constraints $C.4$-$C.7$ signify resource constraints at the device and server levels. Specifically, $C.4$ represents the limit on local computational resources available for each device $m$, ensuring allocations do not surpass the available resources. $C.5$ defines the maximum allowable transmit power for each XR device $m$. $C.6$ ensures that the allocated transmission rate for each device does not exceed the maximum transmission rate. $C.7$ restricts the total computational resource allocation at the MEC server, ensuring that the sum of allocated resources for edge computing remains within the allowable limit. Constraints $C.8$-$C.10$ define the upper bounds on the resource allocations for each queue to avoid over-allocation.

The formulated bi-level optimization problem involves a dual time-scale structure, making it inherently challenging to solve. The lower-level resource allocation is executed in every time slot $t_s$, whereas DNN partition adjustments are performed in the partition adjustment period $t_p$. Furthermore, uncertainties in wireless channel conditions, available computational resources, queue lengths, and task arrival rates introduce significant complexity to the optimization process. Additionally, the upper-level decisions cannot fully observe the dynamics of the lower-level problem, complicating the optimal DNN partitioning strategy.

To effectively address these challenges, we employ a Lyapunov-guided DRL approach in the next section, termed LyaPPO. This approach is designed to efficiently manage the dual time-scale nature of the system, optimize energy consumption, and ensure the stability of the queues across all XR devices.

\section{Lyapunov-guide Proximal Policy Optimization Algorithm}  \label{Sec4}
To address the original bi-level optimization problem, we decompose it into two sub-problems: an upper-level optimization problem and a lower-level optimization problem. In the lower-level optimization problem, the partition decisions of the DNN models are considered predetermined, and the focus is placed on jointly allocating computational and communication resources within the system. To tackle the challenges posed by the long-term queue stability constraints, we first reformulate it using the Lyapunov optimization method and subsequently solve it via convex optimization techniques. In the upper-level optimization problem, we model it as a Markov Decision Process (MDP) and employ the Proximal Policy Optimization (PPO) algorithm to find the optimal solution.

\subsection{Solution to the lower-level problem}
Given the partition point, the problem can be simplified to the following lower-level optimization problem concerning the allocation of communication and computational resources in the system.

\begin{Prob}[Lower-level Optimization Problem] Assuming the partition point $k_{m,n}^{t_p}$ is given, the communication and computational resource allocation problem can be expressed as follows:
\begin{align}
        \underset{\substack{\bm{r}(t), \bm{f}^l(t), \\ \bm{p}(t), \bm{f}^e(t)}}{\textrm{min}}
        & \! \! \!  \lim\limits_{T_s\to+\infty} \sum_{t_s=1}^{T_s}\frac{1}{T_s} [\displaystyle\sum_{{m \in \mathcal{M}}}  \big(E_{m}^l(t_s) + E_{m}^t(t_s)\big)]\nonumber\\
        \quad \textrm{s.t.} 
        &  \qquad C.1 - C.11 \nonumber
    \end{align} 
\vspace{-6mm}
\end{Prob}

Due to the long-term queue stability constraints, this problem remains challenging to solve. Therefore, we employ Lyapunov optimization to transform the problem into a series of single-slot optimization tasks, enabling tractable solutions.

We define the Lyapunov functions $V_{m,n}^l(t_s)$, $V_{m,n}^t(t_s)$ and $V_{m,n}^e(t_s)$ as quadratic functions to measure the congestion levels in the respective queues:
\begin{align}
V^l_{m,n}(t_s) &= \frac{1}{2}(Q^l_{m,n}(t_s))^2, \label{equ:Vl} \\
V^t_{m,n}(t_s) &= \frac{1}{2}(Q^t_{m,n}(t_s))^2, \label{equ:Vt}  \\
V^e_{m,n}(t_s) &= \frac{1}{2}(Q^e_{m,n}(t_s))^2 \label{equ:Ve} ,
\end{align}

Combining with \eqref{Ql_update}, \eqref{Qt_update}, \eqref{Qe_update}, \eqref{equ:Vl}, \eqref{equ:Vt} and \eqref{equ:Ve}, it can be further derived as
\begin{align}
V^l_{m,n}(t_s+1) &= \frac{1}{2}(Q^l_{m,n}(t_s+1))^2 \nonumber \\
&\leq \frac{1}{2}(Q^l_{m,n}(t_s))^2+\frac{1}{2}B^l_{m,n} \nonumber \\ 
&+Q^l_{m,n}(t_s)(a^l_{m,n}(t_s)c^{local}_{m,n}(t_p)-f^{l}_{m,n}(t_s)\tau), \label{ly1}\\
V^t_{m,n}(t_s+1) &= \frac{1}{2}(Q^t_{m,n}(t_s+1))^2 \nonumber \\
&\leq \frac{1}{2}(Q^t_{m,n}(t_s))^2+\frac{1}{2}B^t_{m,n} \nonumber \\ 
&+Q^t_{m,n}(t_s)(a^t_{m,n}(t_s)d_{m,n}(t_p)-r_{m,n}(t_s)\tau), \label{ly2}\\
V^e_{m,n}(t_s+1) &= \frac{1}{2}(Q^e_{m,n}(t_s+1))^2 \nonumber \\
&\leq \frac{1}{2}(Q^e_{m,n}(t_s))^2+\frac{1}{2}B^e_{m,n} \nonumber \\ 
&+Q^e_{m,n}(t_s)(a^e_{m,n}(t_s)c^{edge}_{m,n}(t_p)-f^{e}_{m,n}(t_s)\tau), \label{ly3}
\end{align}
where $B^l_{m,n}= (a^l_{max}c^{local}_{m,n}(t_p))^2+(F^{l}_{m}\tau)^2$, $B^t_{m,n}=(a^t_{max}d_{m,n}(t_p))^2+(R_{m}(t_s)\tau)^2$ and $(a^e_{max}c^{edge}_{m,n}(t_p))^2+(F^e\tau)^2$.

$a^l_{max}$, $a^t_{max}$ and $a^e_{max}$ are maximum number of arriving tasks. Let $\Delta V_{m}^l$, $\Delta V_{m}^t$ and $\Delta V_{m}^e$ to denote the Lyapunov drift, which represents the expected change of the Lyapunov function from the current queues status to the next queue status. Combine with \eqref{ly1}, \eqref{ly2} and \eqref{ly3},  the Lyapunov drift function can be given by
\begin{align}
\Delta V_{m}^l \leq & \sum_{{n \in \mathcal{N}_m}} (V_{m,n}^l(t_s+1)-V_{m,n}^l(t_s)) \nonumber\\
=&  \sum_{{n \in \mathcal{N}_m}} (Q^l_{m,n}(t_s)(a^l_{m,n}(t_s)c^{local}_{m,n}(t_p)-f^{l}_{m,n}(t_s)\tau)\nonumber\\
&+\frac{1}{2}B^l_{m,n}),\\
\Delta V_{m}^t \leq &\sum_{{n \in \mathcal{N}_m}} (V_{m,n}^t(t_s+1)-V_{m,n}^t(t_s)) \nonumber\\
=&  \sum_{{n \in \mathcal{N}_m}} (Q^t_{m,n}(t_s)(a^t_{m,n}(t_s)d_{m,n}(t_p)-r_{m,n}(t_s)\tau)\nonumber\\
&+\frac{1}{2}B^t_{m,n}),\\
\Delta V_{m}^e \leq & \sum_{{n \in \mathcal{N}_m}} (V_{m,n}^e(t_s+1)-V_{m,n}^e(t_s)) \nonumber\\
=&  \sum_{{n \in \mathcal{N}_m}} (Q^e_{m,n}(t_s)(a^e_{m,n}(t_s)c^{edge}_{m,n}(t_p)-f^e_{m,n}(t_s)\tau)\nonumber\\
&+\frac{1}{2}B^e_{m,n}),
\end{align}

An effective way to maintain queue stability is to minimize the Lyapunov drift function. To achieve this goal, we need to optimize both the energy consumption of XR devices and queue stability. Therefore, this problem can be further transformed into minimizing the Lyapunov drift-plus-penalty. The upper bound of the Lyapunov drift-plus-penalty can be expressed as
\begin{align}
&  \sum_{{m \in \mathcal{M}}}(\Delta V_{m}^l+\Delta V_{m}^t+\Delta V_{m}^e+U^l E_m^l(t_s)+U^t E_m^t(t_s)) \nonumber \\
&\leq  \sum_{{m \in \mathcal{M}}} \big( \sum_{{n \in \mathcal{N}_m}} Q^l_{m,n}(t_s)(a^l_{m,n}(t_s)c^{local}_{m,n}(t_p)-f^{l}_{m,n}(t_s)\tau)\nonumber\\
&+ \sum_{{n \in \mathcal{N}_m}}Q^t_{m,n}(t_s)(a^t_{m,n}(t_s)d_{m,n}(t_p)-r_{m,n}(t_s)\tau)\nonumber\\
&+ \sum_{{n \in \mathcal{N}_m}}Q^e_{m,n}(t_s)(a^e_{m,n}(t_s)c^{edge}_{m,n}(t_p)-f^e_{m,n}(t_s)\tau) \nonumber\\
&+ U^l E_{m}^l(t_s) + U^t E_{m}^t(t_s)\big), \label{drift_plus_penalty}
\end{align}
where the $U^l$ and $U^t$ are the weight parameters
to make a trade-off between the minimization of the objective function and the stability of the distributed queues.

Assuming $k_{m,n}^{t_p}$ is given and only the resource allocation for a single time slot needs to be considered, $Q^l_{m,n}(t_s)$, $Q^t_{m,n}(t_s)$, $Q^e_{m,n}(t_s)$, $a^l_{m,n}(t_s)$, $a^t_{m,n}(t_s)$, $a^e_{m,n}(t_s)$, $c^{local}_{m,n}(t_p)$, $d_{m,n}(t_p)$ and $c^{edge}_{m,n}(t_p)$ are determined before resource allocation, and can be ignored. Then, the original problem can be decoupled into the following three distributed sub-problems for solving.

\begin{Prob}[Sub-problem of Local Computing] The sub-problem of local computational resource allocation for XR device $m$ can be expressed as follows:
\label{sub1}
\begin{align}
 \underset{\bm{f}^l(t)}{\textrm{min}} \sum_{{m \in \mathcal{M}}} &\big( \sum_{{n \in \mathcal{N}_m}} -f^{l}_{m,n}(t_s)\tau Q^l_{m,n}(t_s)+ U^l E_{m}^l(t_s)\big),    \nonumber \\
\quad \textrm{s.t.} 
& \qquad C.4, \ C.8. \nonumber
\end{align}
\end{Prob}

\begin{Prob}[Sub-problem of Transmission] The sub-problem of transmission resource allocation for XR device $m$ can be expressed as follows:
\label{sub2}
\begin{align}
\underset{\bm{r}(t),\bm{p}(t)}{\textrm{min}} \sum_{{m \in \mathcal{M}}} & \sum_{{n \in \mathcal{N}_m}}\big(-r_{m,n}(t_s)\tau Q^t_{m,n}(t_s)\big)+ U^t E_{m}^t(t_s) \nonumber \\
\quad \textrm{s.t.} 
& \qquad C.5,\ C.6, \ C.9. \nonumber
\end{align}
\end{Prob}

\begin{Prob}[Sub-problem of Edge Computing] The edge computational resource allocation sub-problem can be formulated as follows:
\label{sub3}
\begin{align}
\underset{\bm{f}^e(t)}{\textrm{min}} & \sum_{{m \in \mathcal{M}}}\sum_{{n \in \mathcal{N}_m}}-f^e_{m,n}(t_s)\tau Q^e_{m,n}(t_s) \nonumber\\
\quad \textrm{s.t.}
& \qquad C.7,\ C.10. \nonumber
\end{align}
\end{Prob}

Then, we solve the three sub-problems separately. For \textbf{Problem} \ref{sub1}, each XR device $m \in \mathcal{M}$ only needs to solve its corresponding sub-problems locally. Considering that \textbf{Problem} \ref{sub1} is a convex problem, we can solve it using CVX \cite{cvx}. For each XR device, the computational complexity of solving this sub-problem is $O(N_m^{3.5})$ \cite{05}. Similarly, \textbf{Problem} \ref{sub2} can also be solved using CVX, and its complexity is $O(N_m^{3.5})$.

For \textbf{Problem} \ref{sub3}, the sub-problem can also be solved by CVX. However, with the computational complexity of $O((MN_m)^{3.5})$, when the number of users increases, the computational time will become unacceptable. To address this, we propose a low-complexity solution method, which is detailed below.

We employ the Lagrange multiplier method to find the optimal solution to \textbf{Problem} \ref{sub3}. The partial Lagrangian function is expressed as follows:
\begin{align}
\label{eq:Lagrangian}
L(f_{m,n}^{e}\left(t_{s}),\lambda,\mu_{m,n},v_{m,n}\right) & \! \! = \! \! \sum_{m}\sum_{n}-f_{m,n}^{e}\left(t_{s}\right)\tau Q_{m,n}^{e}\left(t_{s}\right) \nonumber\\
& \! \! \! \! \! \! \! \! \! \! \! \! \! \! \! \! \! \! \! \! \! \! \! \! \! \! \! \! \! \! \! \! \! \! \! \! \! \! \! \! \! \! \! \! \! \! \! \! \! \! \! \! \! \! \! \! \!  \! \! +   \lambda\left(\sum_{m}\sum_{n}f_{m,n}^{e}(t_{s})-F^{e}\right)^{}-\sum_{m}\sum_{n}\mu_{m,n}f_{m,n}^{e}\left(t_{s}\right)\nonumber\\
& \! \! \! \! \! \! \! \! \! \! \! \! \! \! \! \! \! \! \! \! \! \! \! \! \! \! \! \! \! \! \! \! \! \! +  \sum_{m}\sum_{n}v_{m,n}(f_{m,n}^{e}\left(t_{s}\right)-\frac{Q_{m,n}^{e}\left(t_{s}\right)}{\tau}) .
\end{align}

Using the Lagrangian function, \textbf{Problem} \ref{sub3} can be solved based on the Karush-Kuhn-Tucker (KKT) conditions \cite{20}, as described in the following theorem.

\begin{theorem}
Consider $ \{\tau Q^e_{m,n}\} $, where each element corresponds to a tuple of indices $ (m, n) $. Upon sorting this sequence, an ordered sequence $ f_1 \ge f_2 \ge \cdots \ge f_{(\sum_{m \in \mathcal{M}}N_m)} $ can be obtained. Given a threshold $ K $, if the condition $\sum_{k=1}^{K}f_k < F^e \le \sum_{k=1}^{K+1}f_k$ is met, the optimal solution can be given by
\begin{align}\label{fe}
f_{m,n}^{e*}(t_{s}) =
\begin{cases}
0, & \text{if } k > K, \\
F^e - \sum_{k=1}^{K}f_k , & \text{if } k = K, \\
\frac{Q_{m,n}^{e}(t_{s})}{\tau}, & \text{if } k < K. \\
\end{cases} 
\end{align}  
\end{theorem}

\begin{proof}\renewcommand{\qedsymbol}{}
Please refer to Appendix A.
\end{proof}

\subsection{Solution to the upper-level problem}
To optimize energy consumption for XR devices on a large time scale and ensure queue stability, we define the upper-level problem of selecting the DNN partition points as an infinite-horizon Markov Decision Process (MDP). The MDP is represented by the tuple $\{\mathcal{S}, \mathcal{A}, \mathcal{R}, \gamma_{ppo} \}$. $\mathcal{S}$ denotes the state space, representing the current environment of collaborative inference between XR devices and the MEC server. $\mathcal{A}$ denotes the action space, consisting of the possible DNN model partition points. $\mathcal{R}$ denotes the reward function, and $ \gamma_{ppo}$ is the discount factor where $0 < \gamma_{ppo} <1$. The MDP is defined as follows,

\begin{enumerate}
\item\textbf{State:}
The state $\bm{s}^{t_p} \in \mathcal{S}$ represents the observation during each partition adjustment period. Specifically, it consists of the queue backlogs of the distributed queues $Q_{m,n}^l(t_s)$, $Q_{m,n}^t(t_s)$, $Q_{m,n}^e(t_s)$, the computational complexity and feature map data associated with the current partition point $k_{m,n}^{t_p}$, represented as $c_{m,n}^{local}(t_p)$, $d_{m,n}(t_p)$, $c_{m,n}^{edge}(t_p)$, the computational and communication resources allocated, including $f_{m,n}^l(t_s)$, $p_m(t_s)$, $r_{m,n}(t_s)$, $f_{m,n}^e(t_s)$ and the rate of device $R_m(t_s)$. Therefore, the state $\bm{s}^{t_p} \in \mathcal{S}$ at adjustment period $t_p$ is thus defined as
\begin{align}
\bm{s}^{t_p} = \{&\overline{Q}_{m,n}^l(t_p),\overline{Q}_{m,n}^t(t_p),\overline{Q}_{m,n}^e(t_p),\nonumber\\
&c_{m,n}^{local}(t_p),d_{m,n}(t_p),c_{m,n}^{edge}(t_p),\nonumber\\
&\overline{f}_{m,n}^l(t_p),\overline{p}_m(t_p),\overline{r}_{m,n}(t_p),\overline{f}_{m,n}^e(t_p), \overline{R}_m(t_p), \nonumber\\
&\forall m \in M, \ n \in N_m \},\label{state}
\end{align}
where 
\begin{align}
\overline{Q}_{m,n}^l(t_p) =  & \frac{1}{G}\sum_{t_s=t_pG}^{t_pG+G-1} {Q}_{m,n}^l(t_s), \label{mean1}\\
\overline{Q}_{m,n}^t(t_p) = &  \frac{1}{G} \sum_{t_s=t_pG}^{t_pG+G-1} {Q}_{m,n}^t(t_s), \label{mean2}\\
\overline{Q}_{m,n}^e(t_p) = &  \frac{1}{G} \sum_{t_s=t_pG}^{t_pG+G-1} {Q}_{m,n}^e(t_s), \label{mean3}\\
\overline{f}_{m,n}^l(t_p) =& \frac{1}{G} \sum_{t_s=t_pG}^{t_pG+G-1} {f}_{m,n}^l(t_s), \label{mean4}\\
\overline{p}_{m}(t_p) = & \frac{1}{G} \sum_{t_s=t_pG}^{t_pG+G-1} {p}_{m}(t_s) ,\label{mean5}\\
\overline{r}_{m,n}(t_p) = & \frac{1}{G} \sum_{t_s=t_pG}^{t_pG+G-1} {r}_{m,n}(t_s), \label{mean6}\\
\overline{f}_{m,n}^e(t_p) = & \frac{1}{G} \sum_{t_s=t_pG}^{t_pG+G-1} {f}_{m,n}^e(t_s), \label{mean7} \\
\overline{R}_{m}(t_p) = & \frac{1}{G} \sum_{t_s=t_pG}^{t_pG+G-1} {R}_{m}(t_s).\label{mean8}
\end{align}

\item\textbf{Action Space:} The action space comprises the set of possible partition points for the DNN models on the XR devices. At each partition adjustment period $t_p$, the action is denoted by $\bm{a}^{t_p} = \{a_{m,n}^{t_p}, \forall m,n\} \in \mathcal{A}$, where $a_{m,n}^{t_p} \in \mathcal{K}_{m,n}$ represents the chosen partition point for the DNN model corresponding to service $n$ on device $m$.

\item\textbf{Reward Function:} The reward function is a weighted sum of two parts: XR device's energy consumption and the scaled sum of distributed queue lengths. Therefore, the reward $\bm{r}^{t_p}$ can be given by

\begin{figure}[th]
\includegraphics[width=0.5\textwidth]
{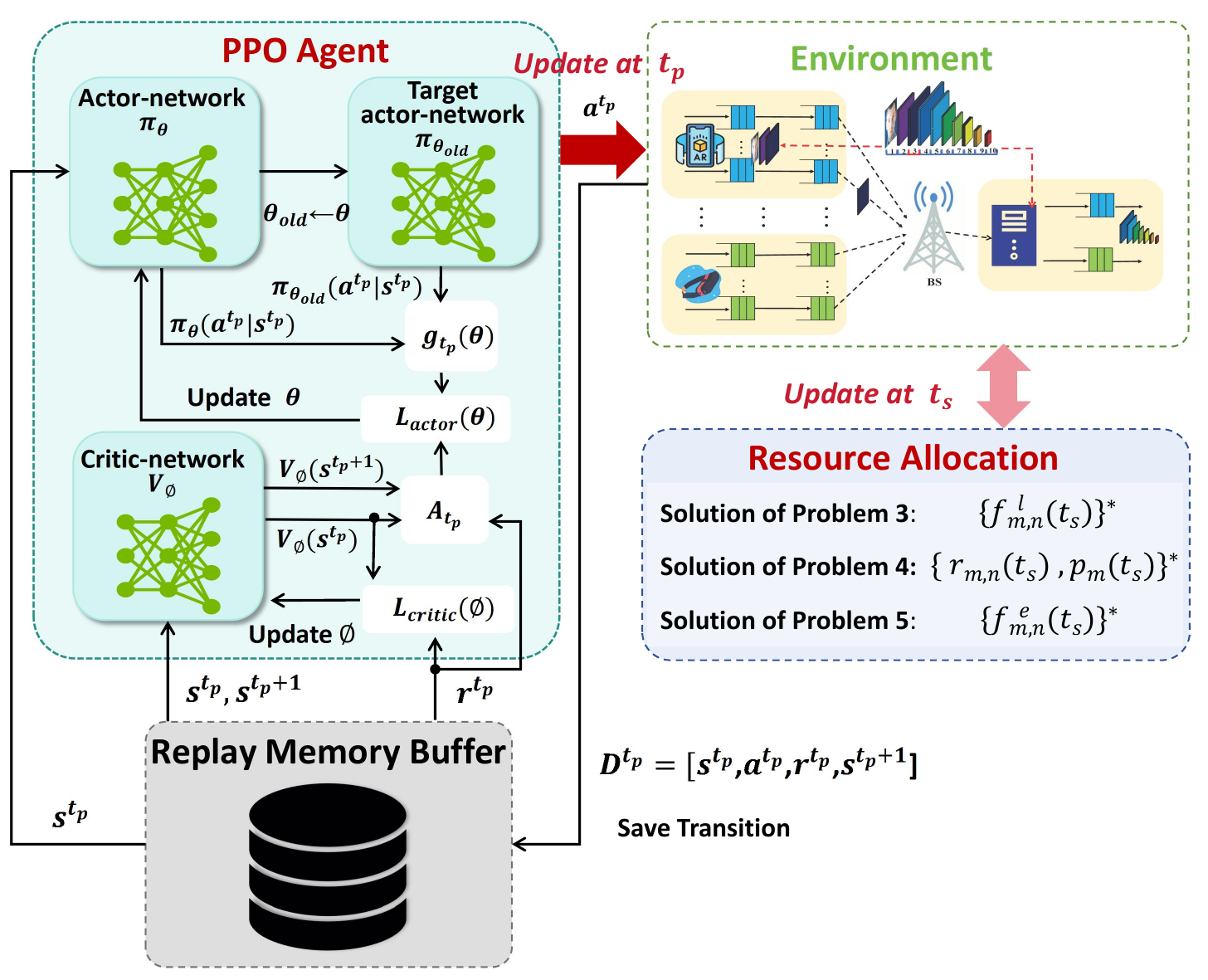}
\caption{The schematics of the LyaPPO algorithm.}
\label{fig-Algo} 
\end{figure}

\begin{align}
    r^{t_p} &= - \sum_{m \in \mathcal{M}}\left[\omega_1(\widetilde{E}^l_m(t_p) + \widetilde{E}^t_m(t_p))+\omega_2 \bar{Q}\right], \label{reward}
\end{align}
where $\omega_1,\omega_2$ denote the weight coefficients. $\bar{Q}$ represents the scaled sum of $Q_{m,n}^l(t_p), Q_{m,n}^t(t_p), Q_{m,n}^e(t_p).$ $\widetilde{E}^l_m(t_p)$ and $\widetilde{E}^t_m(t_p)$ represent the energy consumption caused by local computing and transmission the during the partition adjustment period, and can be respectively expressed as
\begin{align}
\widetilde{E}^l_m(t_p) =  & \frac{1}{G}\sum_{t_s=t_pG}^{t_pG+G-1} {E}_{m,n}^l(t_s), \label{mean_el}\\
\widetilde{E}^t_m(t_p) = &  \frac{1}{G} \sum_{t_s=t_pG}^{t_pG+G-1} {E}_{m,n}^t(t_s).\label{mean_et}
\end{align}

\end{enumerate}

To solve the MDP in the upper-level problem, we employ the PPO algorithm. As shown in Fig. \ref{fig-Algo}, we define $\theta$ and $\theta_{old}$ as the parameters of the actor-network and the target actor-network, respectively. The parameters of the critic network are denoted as $\phi$. In each time slot $t_p$, XR devices and the MEC server perform resource allocation to optimize both queue lengths and energy consumption. During this process, key states, i.e., ${Q}_{m,n}^l(t_s),{Q}_{m,n}^t(t_s),{Q}_{m,n}^e(t_s),{f}_{m,n}^l(t_s),{p}_m(t_s),$ ${r}_{m,n}(t_s),{f}_{m,n}^e(t_s), {h}_m(t_s), E_m^l(t_s), E_m^t(t_s)$, are stored in a temporary memory buffer $\mathcal{D}s$ on a small time scale. During each DNN partition adjustment period $t_p$, the PPO algorithm will calculate rewards using equations \eqref{reward}, \eqref{mean_el}, and \eqref{mean_et}. Subsequently, it will generate the system state based on $\mathcal{D}s$ and select the optimal partition points for all DNN models to ensure efficient task execution and resource utilization. After completing one such cycle, the PPO algorithm records the sampling experience $D_p = \{ \bm{s}^{t_p},\bm{a}^{t_p},r^{t_p}, \bm{s}^{t_p+1} \}$ in a replay memory buffer $\mathcal{D}_p$. This replay buffer is used for model training and will improve policy decisions in subsequent iterations. In model training, the loss function $L_{actor}(\theta)$ of the actor-network can be expressed as
\begin{align}
L_{actor}(\theta) = & \mathbb{E} [ \min ( g_{t_p}(\theta) \cdot A_{t_p}, \nonumber \\
& \text{clip} \left( g_{t_p}(\theta), 1 - \epsilon, 1 + \epsilon \right) \cdot A_{t_p} ) ].  \label{Actorloss}
\end{align}    
where $\mathbb{E}[\cdot]$ denotes the expectation across a mini-batch of samples. $A_{t_p} = r(\bm{s}^{t_p},\bm{a}^{t_p})+ \gamma_{ppo} V_{\phi}(\bm{s}^{t_p+1}) - V_{\phi}(\bm{s}^{t_p})$ and $g_{t_p}(\theta) = \frac{\pi_\theta(\bm{a}^{t_p}|\bm{a}^{t_p})}{\pi_{\theta_{\text{old}}}(\bm{a}^{t_p}|\bm{s}^{t_p})}$ denote the advantage function and the probability ratio between the current and the old policy, respectively. $V_{\phi}(\cdot)$ is the output of the critic-network, which represents the state-value function. The loss function for critic-network $L_{critic}(\phi)$ can expressed as
\begin{equation}
L_{critic}(\phi) = \mathbb{E}_{t} \left[ \left( V_\phi(\bm{s}^{t_p}) - \sum_{k=0}^\infty \gamma_{ppo}^k r_{t_p+k}, \right)^2 \right].
\label{Criticloss}
\end{equation}
In addition, during the training process, the parameters of the target actor-network will be periodically updated to synchronize with those of the actor-network. This synchronization ensures a stable training process and facilitates a gradual improvement in model performance.

Based on the solutions for the upper-level and lower-level problems described above, the pseudocode of our proposed LyaPPO algorithm is presented in Algorithm \ref{algorithm1}.

\begin{algorithm}[h]
    \caption{LyaPPO for Model Partitioning and Resource Allocation}
    \label{algorithm1}
    \KwIn{System parameters in \textbf{Problem 1};}
    \KwOut{Parameter $\theta$ of the convergent actor-network;}
    \BlankLine
    Initialize the actor-network $\pi_{\theta}$ with $\theta$ and the critic-network $V_\phi$ with $\phi$;\\
    Initialize the target actor-network $\pi_{\theta_{old}}$ with $\theta_{old}=\theta$;\\
    Initialize an empty replay memory buffer;
    \BlankLine
    \For{episode = 0 \KwTo max episode}
    {   Initial observation state;\\
        Initial state $\bm{s}^{0}$ and action $\bm{a}^{0}$; \\
        \For{$t_p = 0$ \KwTo $T_p$}
        {

            \For{$t_s = 0$ \KwTo $G-1$}
            {
                Obtain uplink channel gains $h_m(t_s)$ and queue state;\\
                Use CVX to obtain $\{f_{m,n}^l(t_s)\}^{*}$ by solving \eqref{sub1};\\
                Use CVX to obtain $\{r_{m,n}(t_s), p_m(t_s)\}^{*}$ by solving \eqref{sub2};\\
                Calculate \eqref{fe} to obtain $\{f_{m,n}^e(t_s)\}^{*}$;\\
                Store $\{Q_{m,n}^l(t_s), Q_{m,n}^t(t_s), Q_{m,n}^e(t_s)$ $f_{m,n}^l(t_s), p_m(t_s),r_{m,n}(t_s),f_{m,n}^e(t_s),$ $h_m(t_s), E_m^l(t_s), E_m^t(t_s)\}$ in the temporary memory buffer $\mathcal{D}_s$ on a small time scale.
            }
            Calculate reward $r^{t_p}$ by \eqref{reward}-\eqref{mean_et} based on $\mathcal{D}_s$;\\
            Generate $\bm{s}^{t_p+1}$ by \eqref{state}-\eqref{mean8} based on $\mathcal{D}_s$;\\
            Clear the temporary memory buffer $\mathcal{D}_s$; \\
            Store transition $D_p = \{ \bm{s}^{t_p},\bm{a}^{t_p},r^{t_p}, \bm{s}^{t_p+1} \}$ in the replay memory buffer $\mathcal{D}_p$;\\
            Update the state $\bm{s}^{t_p} \gets \bm{s}^{t_p+1}$;\\
            Select the action $\bm{a}^{t_p}$ using $\pi_{\theta_{old}}(\bm{a}^{t_p}|\bm{s}^{t_p})$;\\
        }
            \ForEach{training step}
            {   Sample mini-batch data from the replay memory buffer $\mathcal{D}_p$;\\
                Train actor-network by loss function $L_{actor}(\theta)$; \\
                Train critic-network by loss function $L_{critic}(\phi)$; \\
            }             
            Update target actor-network $\theta_{\text{old}} \gets \theta$;\\
            Clear the replay memory buffer $\mathcal{D}_p$.
    }
\end{algorithm}

\section{Experiment Evaluations} \label{Sec5}

This section provides a comprehensive description of the simulation setup, parameter configurations, and results analysis to evaluate the performance of the proposed LyaPPO algorithm. All simulations were conducted on the PyTorch 1.13.0 platform, utilizing a 13th Gen Intel(R) Core(TM) i9-13900K CPU (3GHz) with 125.5GB of RAM. The uplink bandwidth was set to $B_w = 1$ MHz, and the power spectral density was fixed at $N_0 = -174$ dBm/Hz. Following the channel model in \cite{01}, the average channel gain for each XR device $m \in M$ is expressed as $\bar{h}_m = \left( A_d \frac{3 \times 10^8}{4 \pi f_c d_m} \right)^{d_e}$, where $A_d = 3$ represents the antenna gain, $f_c = 915$ MHz denotes the carrier frequency, $d_e = 3$ is the path loss exponent, and $d_m \in [150, 250]$ m corresponds to the distance to the MEC server. Uplink channels $h_m$ adhere to a Rayleigh fading model, given by $h_m = \beta \bar{h}_m$, where $\beta$ follows an independent exponential distribution with a mean of one.

The simulation setup involves four XR devices ($M = 4$), each running two DNN models ($N_m = 2$) to handle AI service inference tasks. These tasks are generated according to a Poisson arrival process with a mean arrival rate of $\mathbb{E}[a_{m,n}(t_s)] = \lambda$. The computational complexity coefficient $\rho$ is set to 0.12 cycle/MAC as per \cite{05}. Table \ref{table1} summarizes the default simulation parameters. For training, the PPO framework employs two hidden layers with 128 neurons in both the actor and critic networks, a learning rate of $3 \times 10^{-4}$, and the Adam optimizer. The training process is conducted over 2500 episodes, with each episode comprising 200 exploration steps.

\begin{table}[t]
\centering
\footnotesize
\caption{Default Simulation Parameters}
\label{table1}
\renewcommand\arraystretch{1.2}
\begin{tabular}{|c||c|} 
\hline
$(M,N_m)=(4,2)$
&$\lambda = 0.2$ requests/second \\
\hline
$\tau = 10$ ms
&$G = 10$\\
\hline
$(F_m^l,F^e) = (1.5,20)$ GHz 
&$p_m^{\text{max}} = 0.3$ W \\
\hline
$\rho = 0.12$ cycles/MAC 
& $\delta = 10^{-28}$ Watt$\cdot s^3$ \\
\hline
$(V_1,V_2)=(1e9,1e6)$
& $(\beta_1,\beta_2,\beta_3)\!\!\! = (0.6,0.2,0.2)$ \\
\hline
\end{tabular}
\end{table}

\subsection{Experiment Setting}
The proposed LyaPPO algorithm was compared against several baseline schemes.
\begin{enumerate}
    \item \textit{PPO:} Uses PPO to jointly optimize model partitioning and resource allocation \cite{35}. 
    \item \textit{DDPG-Cov:} Employs the deep deterministic policy gradient (DDPG) algorithm for model partitioning and convex optimization for resource allocation \cite{31}.
    \item \textit{Random-Cov:} Selects model partition points randomly within the action space, with convex optimization handling resource allocation \cite{30}.
    \item \textit{Fix-Cov:} Fixes model partition points a priori and allocates resources via convex optimization \cite{04}.
\end{enumerate}

\subsection{Performance Comparison}
\begin{figure}[h]
\includegraphics[width=8cm,height=6cm]{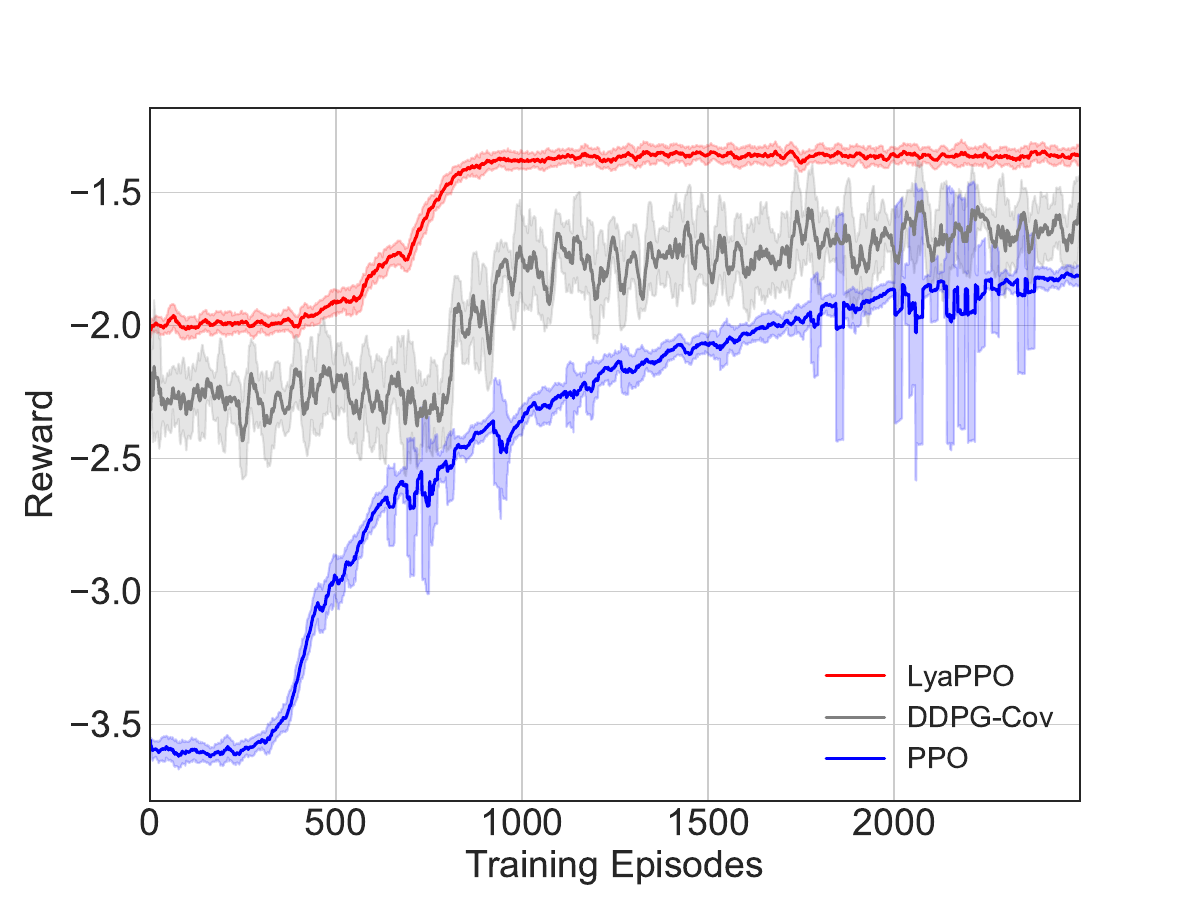}
\centering
\caption{\label{fig:convergence}Comparison of convergence in the training process}
\end{figure}
To compare the convergence performance of the LyaPPO algorithm with other DRL-based baselines, training curves were plotted and smoothed using a sliding window to highlight overall trends in the raw data. Fig. \ref{fig:convergence} shows that LyaPPO achieves the highest average cumulative return and converges around 1000 episodes. 
Compared to the DDPG-Cov baseline, the LyaPPO algorithm shows reduced fluctuations in the convergence curve. The LyaPPO algorithm converges faster than PPO baseline. The main reason is that the LyaPPO algorithm avoids the exploration of high-dimensional action space compared to PPO baseline which uses the DRL framework to jointly adjust model partition points and resource allocation. Therefore, the LyaPPO algorithm achieves the best performance.

\begin{figure}[h]
    \centering
    \captionsetup[subfigure]{justification=centering, font=small, skip=2pt} 
    \begin{subfigure}[t]{\linewidth}
        \centering
        \includegraphics[width=8cm,height=4.3cm]{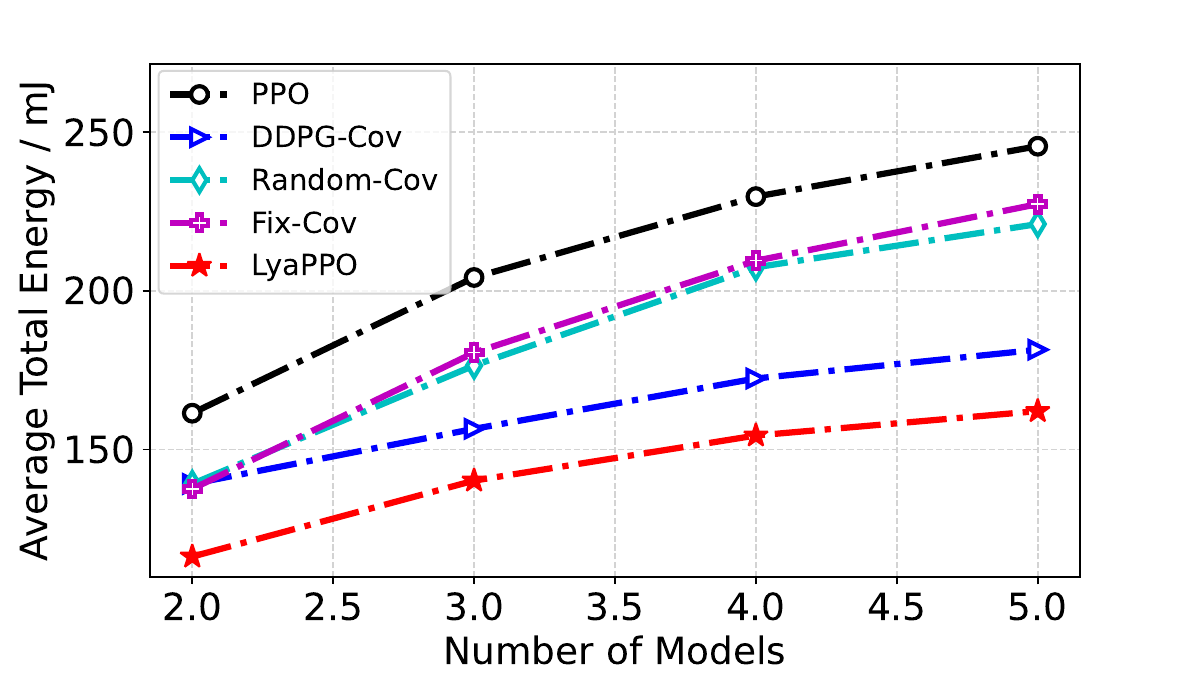}
        \caption{\label{fig:exp24Etotal}Total Energy Consumption}
    \end{subfigure}
    \begin{subfigure}[t]{\linewidth}
        \centering
        \includegraphics[width=8cm,height=4.3cm]{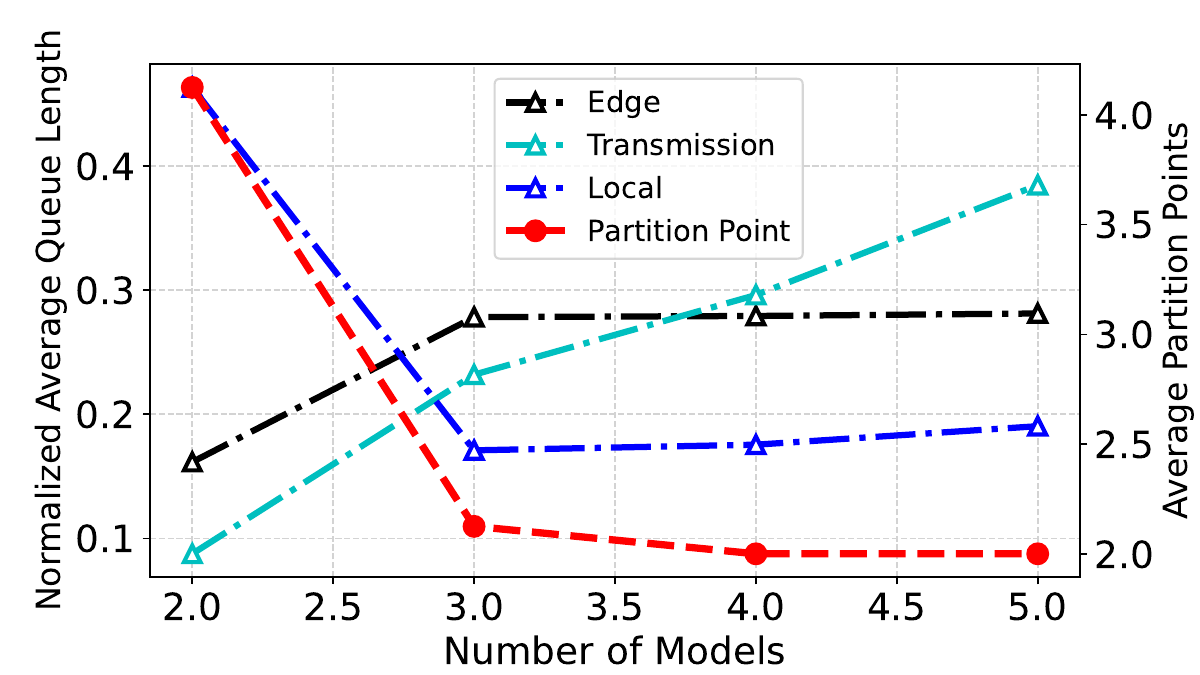}
        \caption{The normalized average queue length and average model partition points of the proposed algorithm\label{fig:exp24_k_queue}}
    \end{subfigure}
    \caption{Comparison of algorithm performance with different numbers of deployed models $N_m$ on the XR device.\label{fig:exp24}}
\end{figure}

In Fig. \ref{fig:exp24}, we evaluate the effect of the number of deployed models on algorithm performance. As demonstrated in Fig. \ref{fig:exp24Etotal}, the average total energy consumption of the LyaPPO algorithm has reduced by 11.98\%, 31.49\%, 22.28\%, 23.18\%, when compared to PPO, DDPG-Cov, Random-Cov, Fix-Cov baselines, respectively. The total energy consumption shows an upward trend with $N_m$ due to the increase in the total number of DNN inference requests on each XR device. To process all the tasks in the queues, XR devices then require higher local computational frequencies for computing, as well as greater transmit powers and transmission rates for offloading tasks to the MEC server.

In Fig. \ref{fig:exp24_k_queue}, the trade-off between local energy consumption and transmission energy consumption is reflected in the adjustment of model partition points and the queue backlogs. As $N_m$ increases, the partition points advance, while the backlogs in local queues diminish, and the backlogs in transmission and edge queues expand.

\begin{figure}[t]
    \centering
    \captionsetup[subfigure]{justification=centering, font=small, skip=2pt} 
    \begin{subfigure}[t]{1\linewidth}
        \centering
        \includegraphics[width=8cm,height=4.3cm]{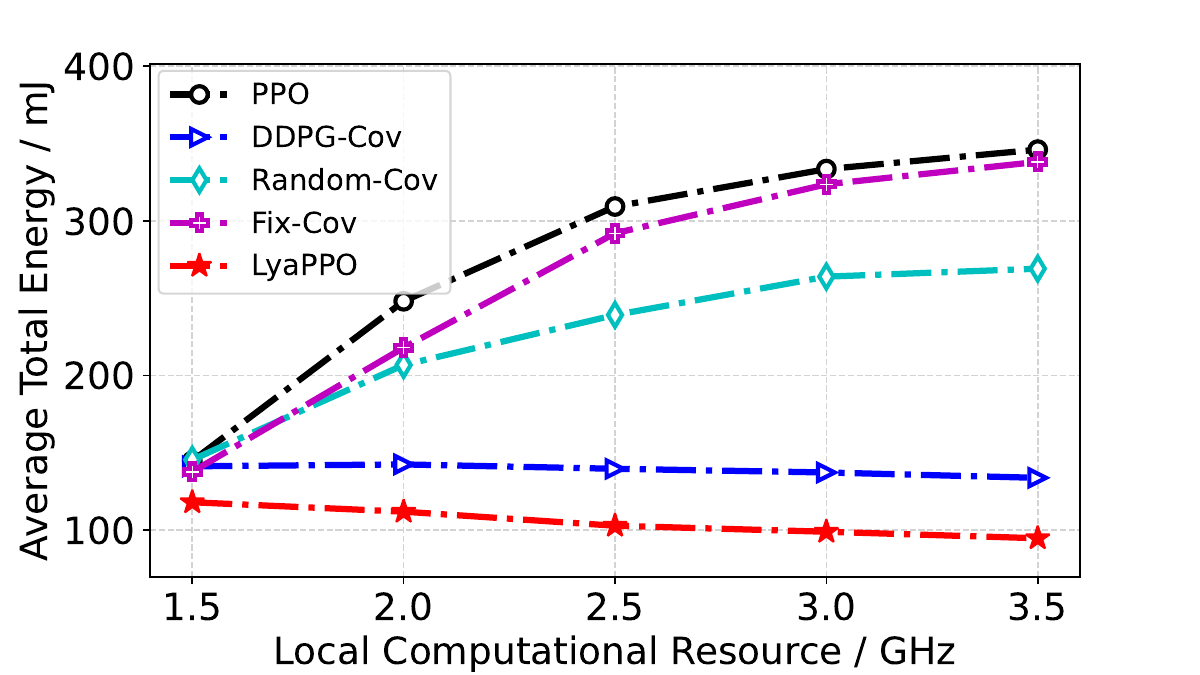}
        \caption{\label{fig:exp22Etotal}Total Energy Consumption}
    \end{subfigure}
    \hfill
    \begin{subfigure}[t]{1\linewidth}
        \centering
        \includegraphics[width=8cm,height=4.3cm]{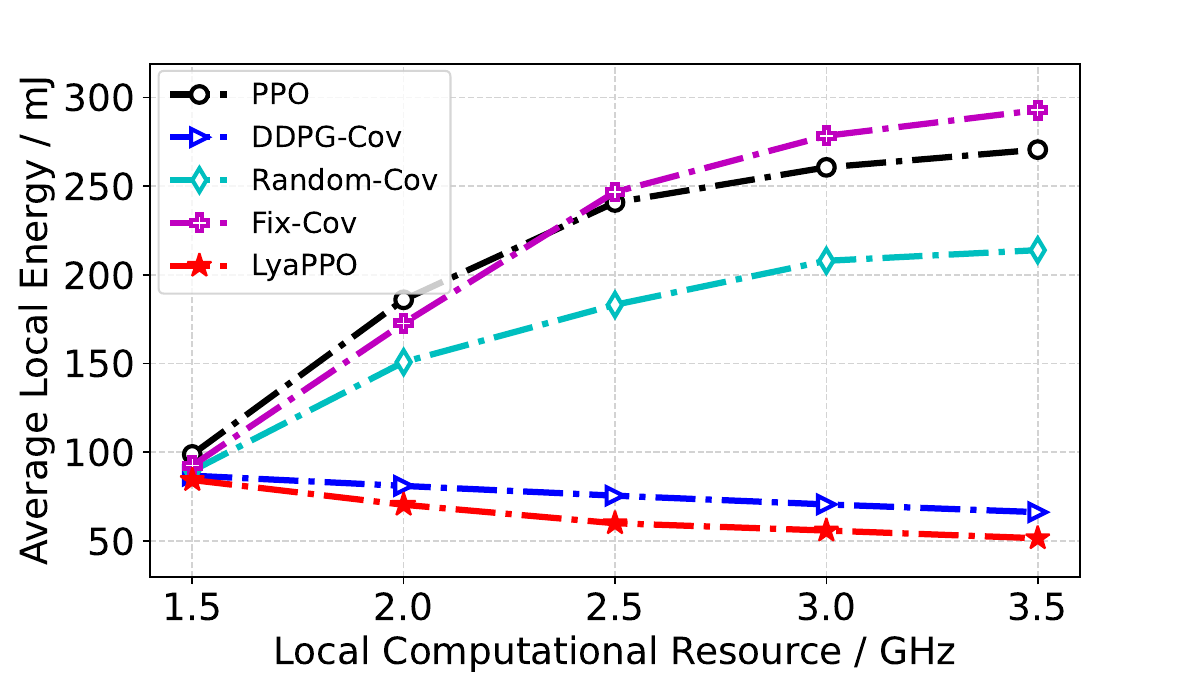}
        \caption{\label{fig:exp22El}Local Energy Consumption}
    \end{subfigure}
    \hfill
    \begin{subfigure}[t]{1\linewidth}
        \centering
        \includegraphics[width=8cm,height=4.3cm]{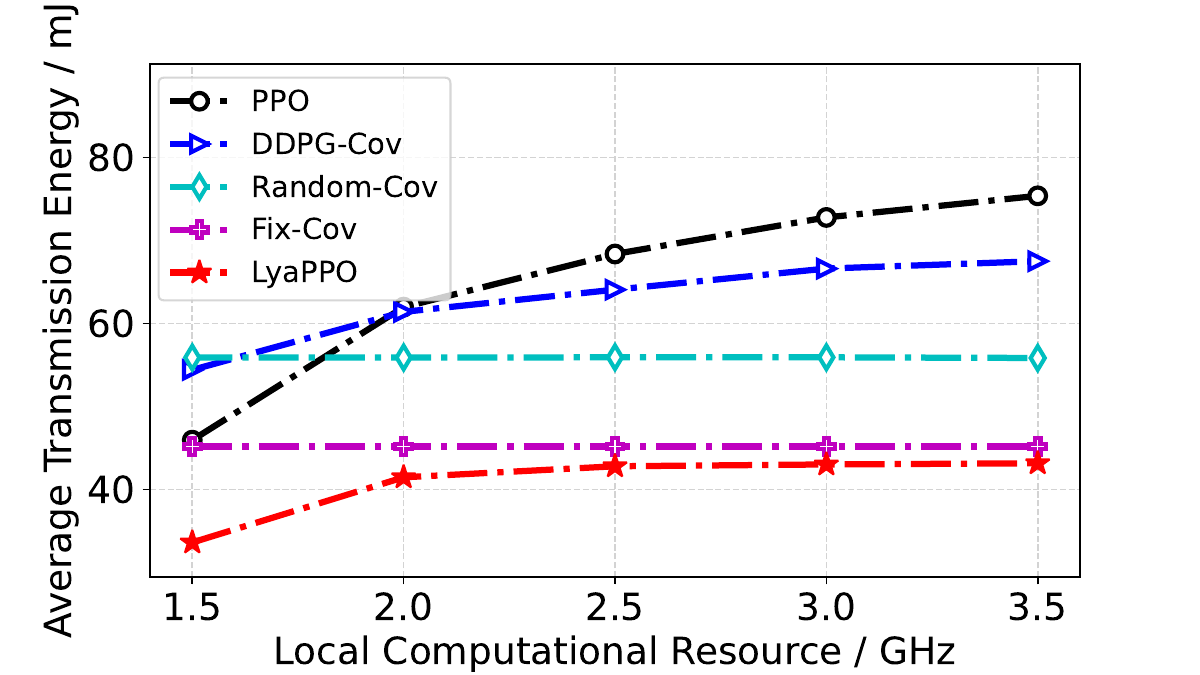}
        \caption{\label{fig:exp22Et}Transmission Energy Consumption}
    \end{subfigure}
    \begin{subfigure}[t]{1\linewidth}
        \centering
        \includegraphics[width=8cm,height=4.3cm]{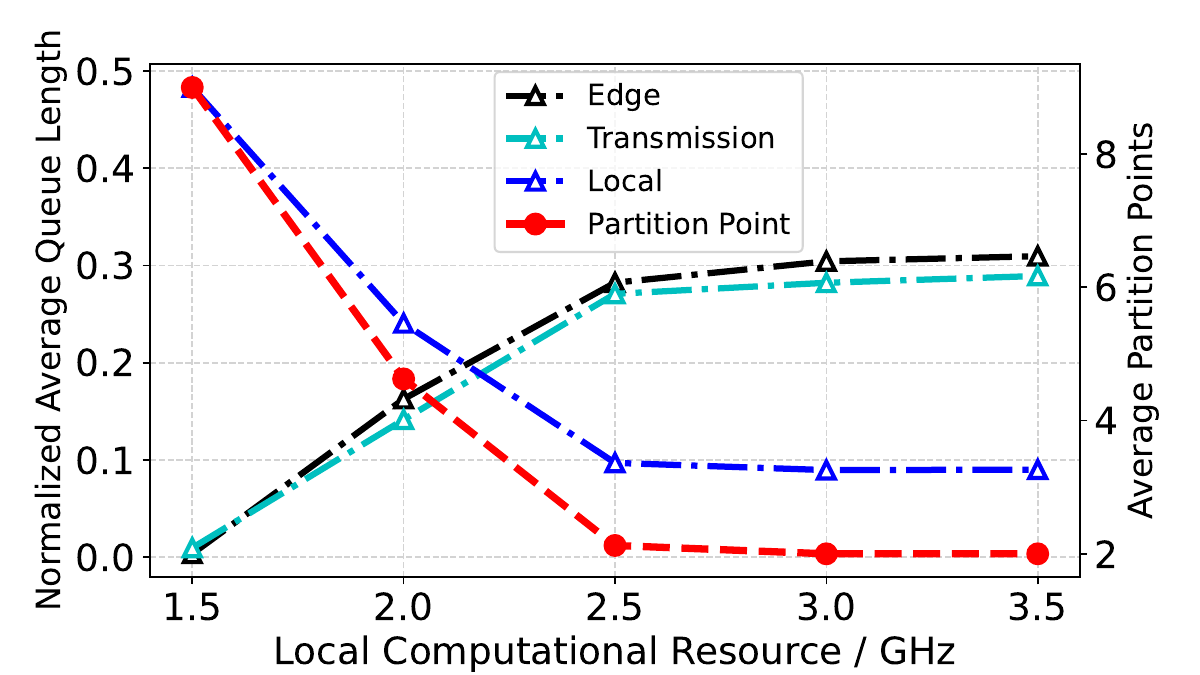}
        \caption{The normalized average queue length and average model partition points of the proposed algorithm\label{fig:exp22qk}   }
    \end{subfigure}
    \caption{Comparison of algorithm performance under varying maximum local computational capacities $F_m^l$ of XR device.\label{fig:exp22}}
\end{figure}

\begin{figure}[ht]
    \centering
    \captionsetup[subfigure]{justification=centering, font=small, skip=2pt} 
    \begin{subfigure}[t]{1\linewidth}
        \centering
        \includegraphics[width=8cm,height=4.3cm]{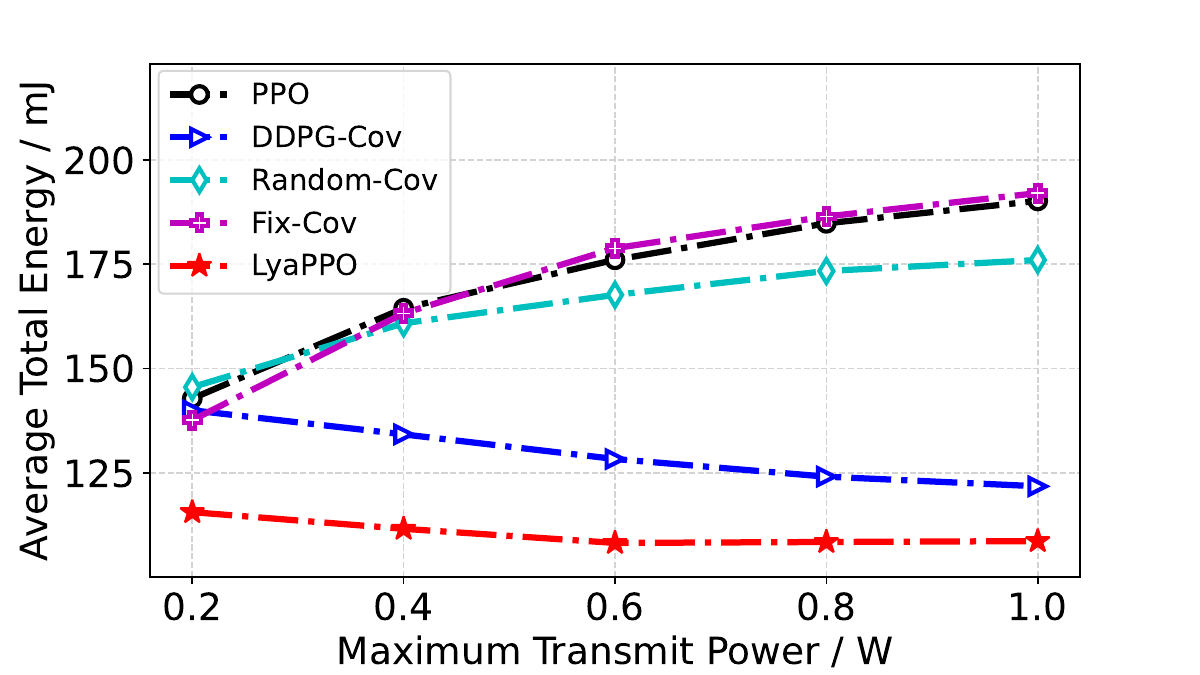}
        \caption{\label{fig:exp23Etotal}Total Energy Consumption}
    \end{subfigure}
    \hfill
    \begin{subfigure}[t]{1\linewidth}
        \centering
        \includegraphics[width=8cm,height=4.3cm]{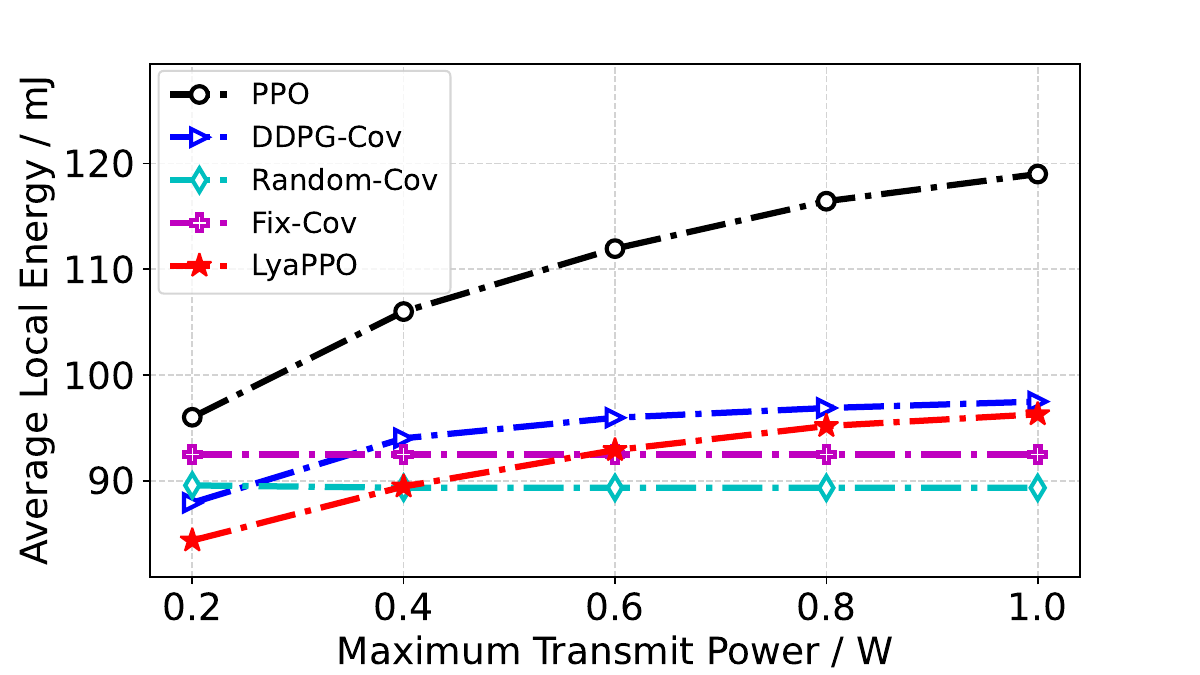}
        \caption{\label{fig:exp23El}Local Energy Consumption}
    \end{subfigure}
    \hfill
    \begin{subfigure}[t]{1\linewidth}
        \centering
        \includegraphics[width=8cm,height=4.3cm]{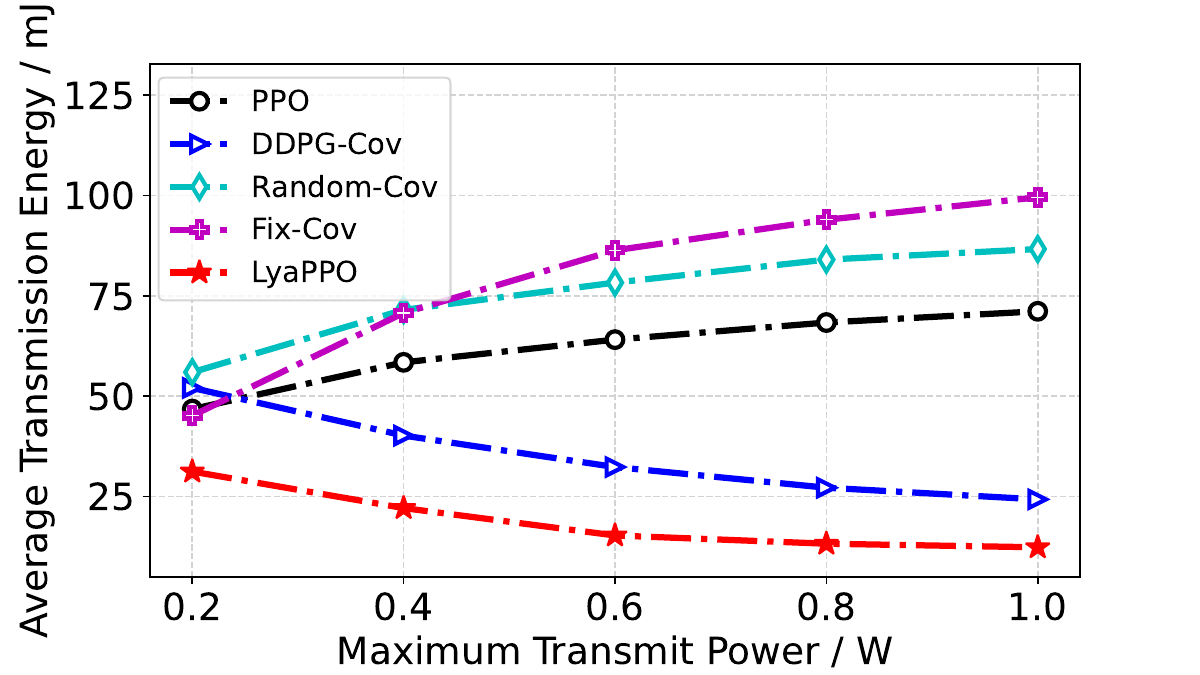}
        \caption{\label{fig:exp23Et}Transmission Energy Consumption}
    \end{subfigure}
    \begin{subfigure}[t]{1\linewidth}
        \centering
        \includegraphics[width=8cm,height=4.3cm]{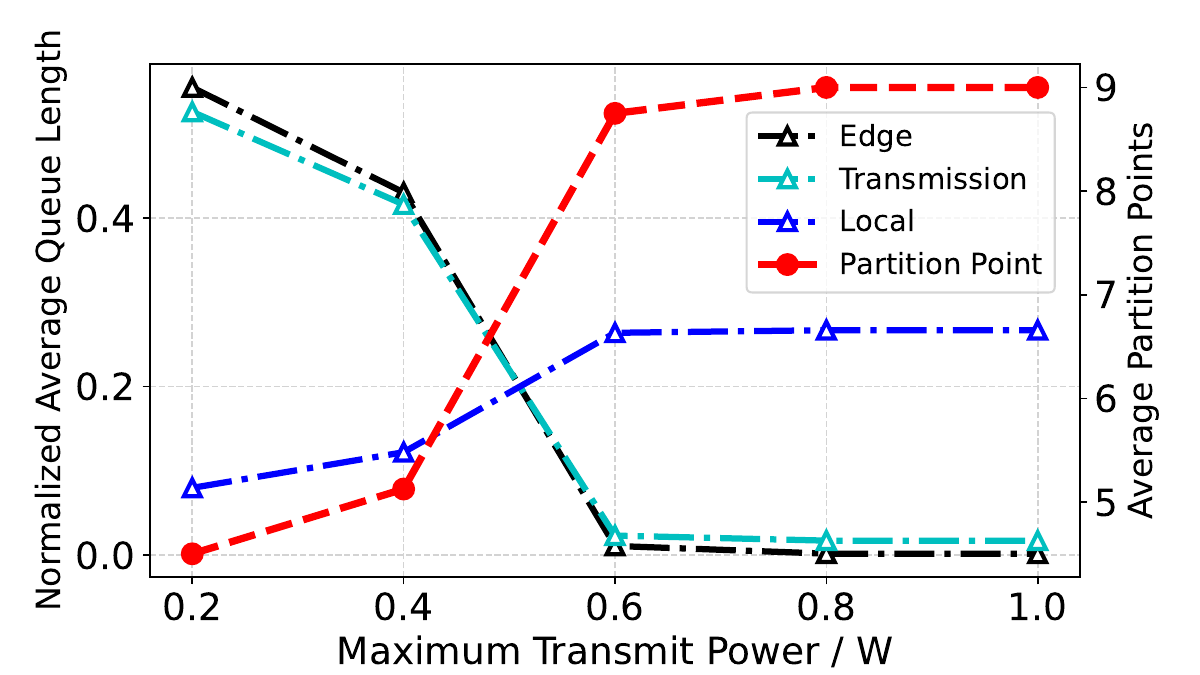}
        \caption{\label{fig:exp23qk}   The normalized average queue length and average model partition points of the proposed algorithm}
    \end{subfigure}
    \caption{Comparison of algorithm performance under varying maximum transmit power $p_m^{max}$ of XR device. \label{fig:exp23}}
   
\end{figure}

\begin{figure}[ht]
    \centering
    \captionsetup[subfigure]{justification=centering, font=small, skip=2pt} 
    \begin{subfigure}[t]{1\linewidth}
        \centering
        \includegraphics[width=8cm,height=4.3cm]{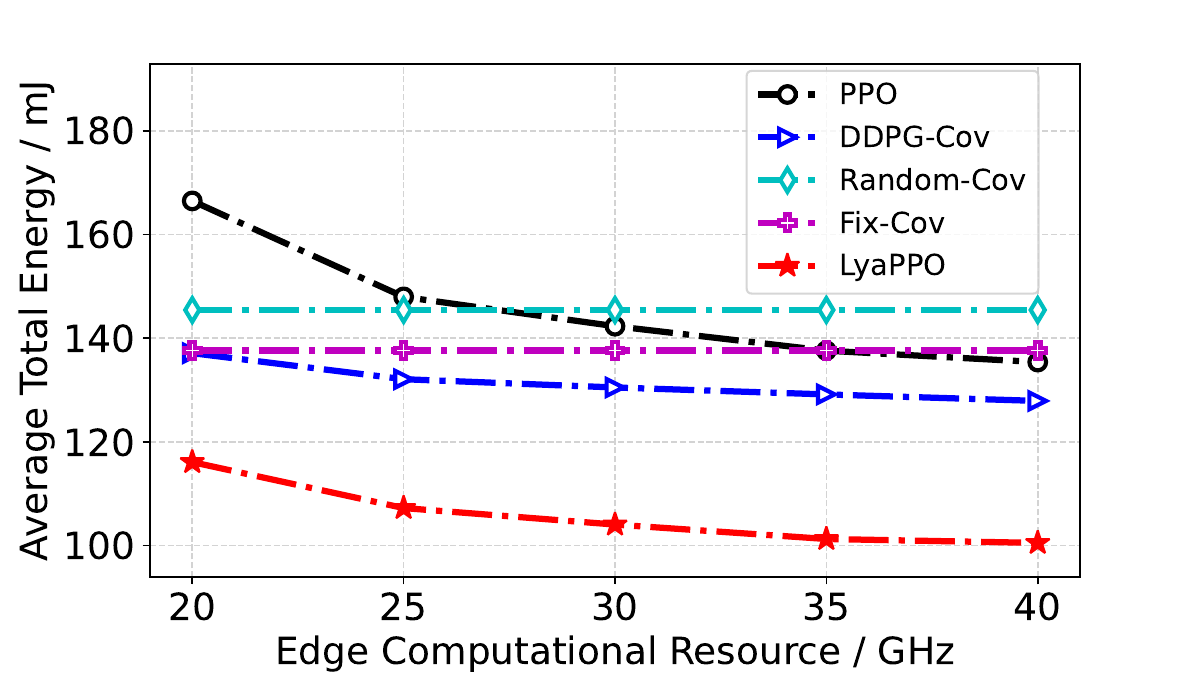}
        \caption{\label{fig:exp21Etotal}Total Energy Consumption}
    \end{subfigure}
    \hfill
    \begin{subfigure}[t]{1\linewidth}
        \centering
        \includegraphics[width=8cm,height=4.3cm]{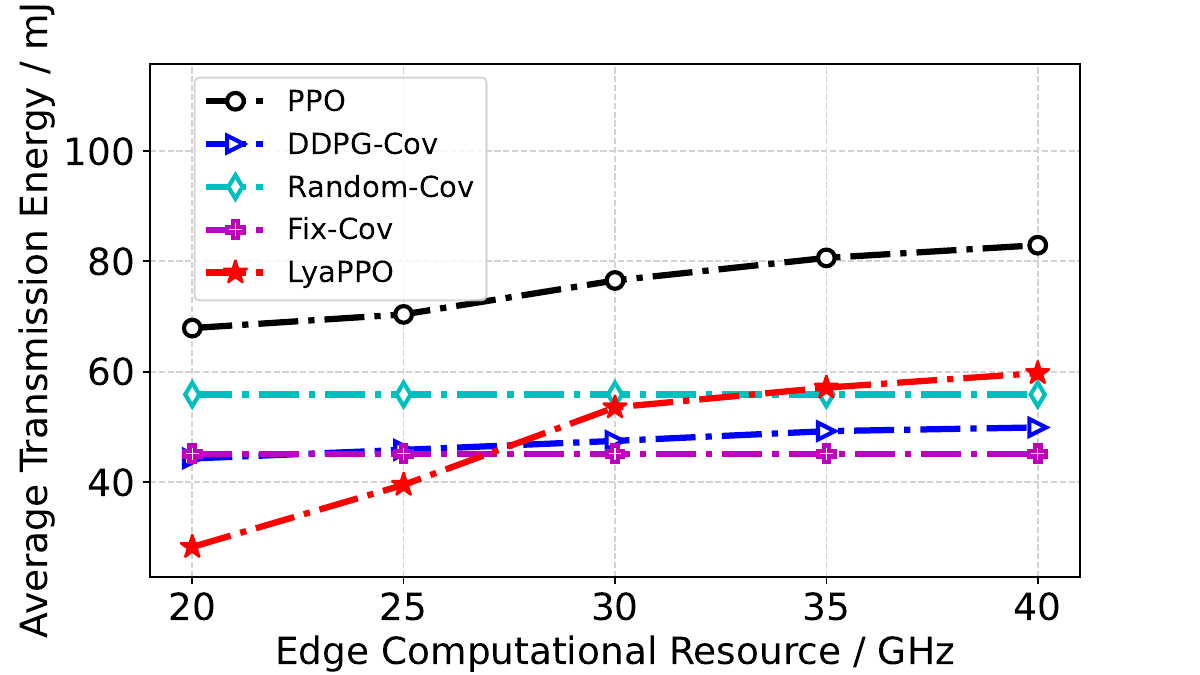}
        \caption{\label{fig:exp21El}Local Energy Consumption}
    \end{subfigure}
    \hfill
    \begin{subfigure}[t]{1\linewidth}
        \centering
        \includegraphics[width=8cm,height=4.3cm]{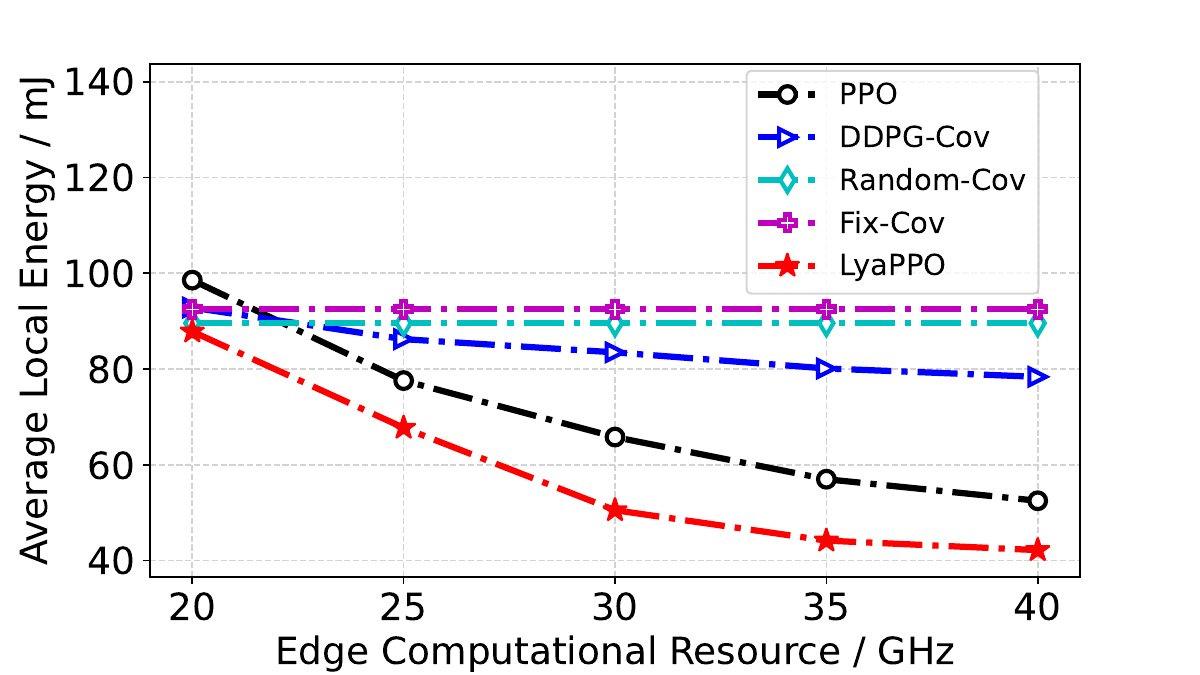}
        \caption{\label{fig:exp21Et}Transmission Energy Consumption}
    \end{subfigure}
    \begin{subfigure}[t]{1\linewidth}
        \centering
        \includegraphics[width=8cm,height=4.3cm]{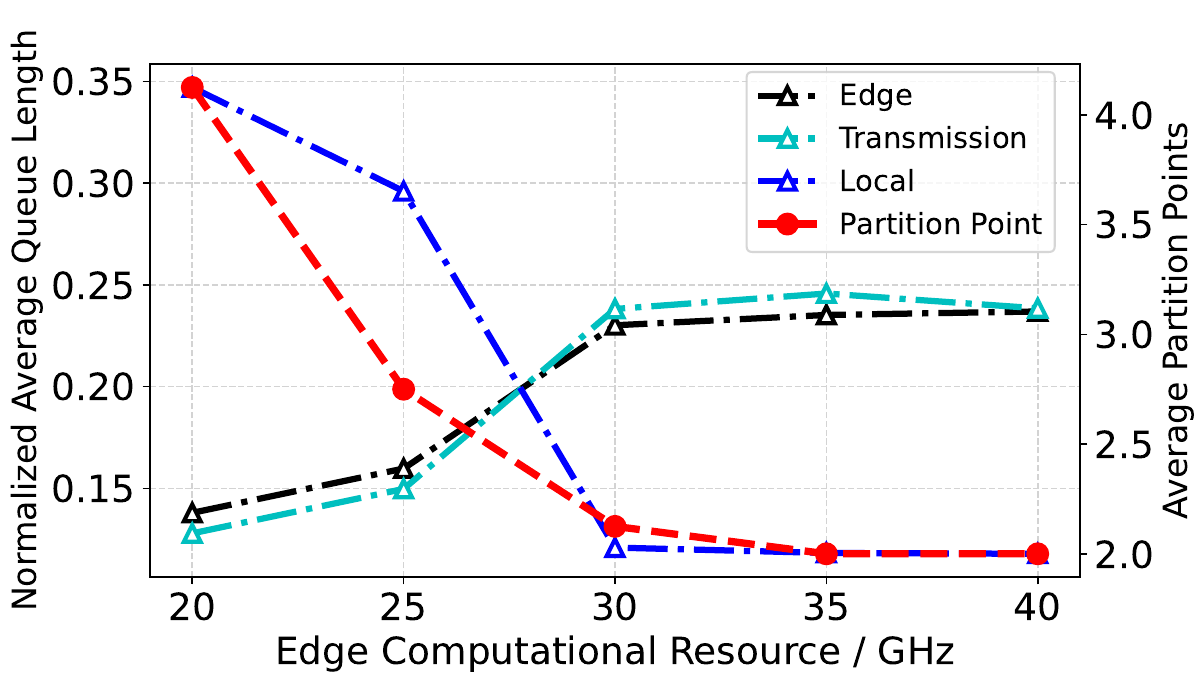}
        \caption{The normalized average queue length and average model partition points of the proposed algorithm\label{fig:exp21qk}}
    \end{subfigure}
    \caption{Comparison of algorithm performance under varying maximum edge computational capacities $F^e$ of MEC server.}
    \label{fig:exp21}
\end{figure}

\subsection{Performance Comparison under Different Resources}
\subsubsection{Impact of the local computational resource}

Fig. \ref{fig:exp22} evaluates the impact of maximum local computational resources on performance. As shown in Fig. \ref{fig:exp22Etotal}, the LyaPPO algorithm achieves a reduction in average total energy consumption by 24.29\%, 56.62\%, 49.82\%, and 53.83\%, compared to the PPO, DDPG-Cov, Random-Cov, and Fix-Cov baselines, respectively. A decreasing trend in total energy consumption can be observed for both the LyaPPO algorithm and the DDPG-Cov baseline. When local computational resources are limited, local computing dominates the process of inference tasks, and the local energy consumption remains acceptable. However, the insufficient local computational resources cannot clear all tasks in the local queues, leading to task backlogs. 

Figs. \ref{fig:exp22El} and \ref{fig:exp22Et} indicate that as the local computational resource capacity increases, the LyaPPO algorithm tends to reduce local computing while increasing transmission, thereby sacrificing part of the transmission energy consumption to achieve greater reductions in local energy consumption. 
Additionally, the adaptive resource allocation capability of XR devices enhances with $F_m^l$, enabling the timely clearing of local queues and reducing average local energy consumption. As shown in Fig. \ref{fig:exp22qk}, the trade-off between local energy consumption and transmission energy consumption is reflected in adjustment of model partition points and queue lengths of distributed queues. As $F_m^l$ increases, partition points gradually shift forward, reducing local queue lengths while transmission and edge queue lengths grow.

\subsubsection{Impact of the maximum transmission power}

Fig. \ref{fig:exp23} evaluates the impact of maximum transmit power on performance. Fig. \ref{fig:exp23Etotal} highlights the superior performance of the LyaPPO algorithm, which reduces average total energy consumption by 14.67\%, 34.78\%, 32.47\%, and 34.48\% compared to the PPO, DDPG-Cov, Random-Cov, and Fix-Cov baselines, respectively. The total energy consumption of the LyaPPO algorithm and the DDPG-Cov baseline decreases with maximum transmit power $p_m^{max}$. 

When the maximum transmit power is low, the maximum transmission rate $R_m(t_s)$ is limited, which results in insufficient allocated transmission rates to clear all tasks in the transmission queues, leading to backlogs. Figs. \ref{fig:exp23El} and \ref{fig:exp23Et} indicate that as the maximum transmit power increases, the LyaPPO algorithm tends to reduce transmission and increase local computing, sacrificing some local energy consumption to conserve greater transmission energy consumption. 

Besides, the XR device can allocate communication resources more adaptively as maximum transmit power increases, enabling timely clearing of transmission queues and reducing average transmission energy consumption. 
As shown in Fig. \ref{fig:exp23qk}, energy trade-off leads to adjustments in model partition points and queue lengths in distributed queues. As $p_m^{max}$ increases, partition points shift back, lengthening local queues while shortening transmission and edge queues.

\subsubsection{Impact of the edge computational resource}

Fig. \ref{fig:exp21} evaluate the impact of edge computational resource on performance. As shown in Fig. \ref{fig:exp21Etotal}, the total energy consumption decreases with edge computational resources for the LyaPPO, DDPG-Cov and PPO algorithms. Compared to the PPO, DDPG-Cov, Random-Cov, and Fix-Cov baselines, the LyaPPO algorithm reduces the average total energy consumption by 19.47\%, 27.36\%, 27.22\%, and 23.09\%, respectively. 
As illustrated in  Figs. \ref{fig:exp21El} and \ref{fig:exp21Et}, when edge computational resource capacity is limited, local energy consumption remains acceptable. This leads to local computing predominantly handling DNN inference tasks. As the edge computational resource capacity increases, the LyaPPO algorithm strategically transmits more inference tasks to MEC server. This approach then effectively optimizes overall energy consumption by sacrificing some transmission energy consumption to save more local energy consumption. Figs. \ref{fig:exp21qk} reflects the impact of energy trade-off, manifested in the adjustments of model partitioning and queue lengths in distributed queues. The local queue shortens, while transmission and edge queues grow, accompanied by a forward shift in the model partition points.

\section{Conclusion} \label{Sec6}
This paper considers the problem of energy optimization for collaborative inference between MEC and XR devices. We propose a multi-task coupled distributed queue model and formulate the problem as a dual time-scale bi-level optimization problem with constraints on computational and communication resources, as well as queue stability. To tackle this problem, we propose a LyaPPO algorithm that updates model partition points in each partition adjustment period using a DRL framework, and adjusts the resource allocation scheme by solving decoupled sub-problems with convex optimization in each time slot. Experiments demonstrate the effectiveness of the LyaPPO algorithm, which outperforms baselines in energy consumption under varying local computational resources, maximum transmit power, and edge computational resources.

\appendices
\section{Proof of Theorem 1}
Based on the KKT conditions \cite{36}, the necessary and sufficient conditions from \eqref{eq:Lagrangian} are as follows
\begin{align}
& \frac{\partial L}{\partial f_{m,n}^{e*}(t_s)}=-\tau Q_{m,n}^e(t_s)+\lambda^*-\mu_{m,n}^*+v_{m,n}^* = 0, \label{kkt1}\\
& \lambda(\sum_{m}\sum_{n}f_{m,n}^{e*}(t_s)-F^e)=0,  \label{kkt2}\\
& \mu_{m,n}^*f_{m,n}^{e*}(t_s) =0 , \mu_{m,n}^* \geq 0, \label{kkt3}\\
& v_{m,n}^*(f_{m,n}^{e*}(t_s)-\frac{Q_{m,n}^{e}(t_{s})}{\tau})=0,v_{m,n}^* \geq 0. \label{kkt4}
\end{align}
where $f_{m,n}^{e*}(t_s)$, $\lambda^*, \mu_{m,n}^*$, and $v_{m,n}^*$ represent  the optimal solution. When $ 0 < f_{m,n}^{e*}(t_s) < \frac{Q_{m,n}^{e}(t_s)}{\tau} $, we can obtain $\mu_{m,n}^* = 0 $ and $ v_{m,n}^* = 0 $ from \eqref{kkt3}-\eqref{kkt4}, leading to $ \lambda^* = \tau Q_{m,n}^e(t_s) $. If $ f_{m,n}^{e*}(t_s) = 0 $, then $ v_{m,n}^* = 0 $ and $ \mu_{m,n}^* \in [0, +\infty) $ based on \eqref{kkt3}-\eqref{kkt4}, implying $ \lambda^* \geq \tau Q_{m,n}^e(t_s) $. Similarly, when $ f_{m,n}^{e*}(t_s) = \frac{Q_{m,n}^{e}(t_s)}{\tau} $, we have $ v_{m,n}^* \in [0, +\infty) $ and $ \mu_{m,n}^* = 0 $, implying $\lambda^* \leq \tau Q_{m,n}^e(t_s)$.

Thus, the optimal solution with respect to $\lambda ^*$ can be expressed as
\begin{align}
f_{m,n}^{e*}(t_{s}) =
\begin{cases}
0, & \text{if } \lambda ^* \ge \tau Q_{m,n}^e(t_s), \\
F^e - \! \! \! \! \displaystyle {\sum_{(m,n)\in \mathcal{\hat{P}}_{m,n}}} \! \! \! \! f_{m',n'}^{e*}(t_{s})  , & \text{if } \lambda ^* = \tau Q_{m,n}^e(t_s), \\
\frac{Q_{m,n}^{e}(t_{s})}{\tau}, & \text{if } \lambda ^* < \tau Q_{m,n}^e(t_s), \\ 
\end{cases} \nonumber
\end{align} 
where $\mathcal{\hat{P}}_{m,n}=\{(m',n')| m' \in \mathcal{M}, n' \in \mathcal{N}_m, (m',n') \neq (m,n) \}$.

When $\lambda ^* >0$, to simplify the optimal solution expression, we sort $\{\tau Q^e_{m,n}\} $ obtaining the ordered sequence $f_1 \ge f_2 \ge \cdots \ge f_{(\sum_{m \in \mathcal{M}}N_m)}$, and the expression of results can be rewritten as
\begin{align}\label{eq:P5_result}
f_{m,n}^{e*}(t_{s}) =
\begin{cases} 
0, & \text{if } k > K, \\
F^e - \sum_{k=1}^{K}f_k , & \text{if } k = K, \\
\frac{Q_{m,n}^{e}(t_{s})}{\tau}, & \text{if } k < K, \\
\end{cases} 
\end{align} 
where $K$ is determined by $\sum_{k=1}^{K}f_k < F^e \le \sum_{k=1}^{K+1}f_k$. When $\lambda ^* =0$, the result expression also matches \eqref{eq:P5_result}, which completes the proof.

\bibliographystyle{IEEEtran}

\bibliography{ref}

\begin{thebibliography}{10}
\providecommand{\url}[1]{#1}
\csname url@samestyle\endcsname
\providecommand{\newblock}{\relax}
\providecommand{\bibinfo}[2]{#2}
\providecommand{\BIBentrySTDinterwordspacing}{\spaceskip=0pt\relax}
\providecommand{\BIBentryALTinterwordstretchfactor}{4}
\providecommand{\BIBentryALTinterwordspacing}{\spaceskip=\fontdimen2\font plus
\BIBentryALTinterwordstretchfactor\fontdimen3\font minus \fontdimen4\font\relax}
\providecommand{\BIBforeignlanguage}[2]{{%
\expandafter\ifx\csname l@#1\endcsname\relax
\typeout{** WARNING: IEEEtran.bst: No hyphenation pattern has been}%
\typeout{** loaded for the language `#1'. Using the pattern for}%
\typeout{** the default language instead.}%
\else
\language=\csname l@#1\endcsname
\fi
#2}}
\providecommand{\BIBdecl}{\relax}
\BIBdecl

\bibitem{t5}
F.~Tang, X.~Chen, M.~Zhao, and N.~Kato, ``The roadmap of communication and networking in {6G} for the metaverse,'' \emph{IEEE Wireless Commun.}, vol.~30, no.~4, pp. 72--81, 2023.

\bibitem{t1}
P.~Hande \emph{et~al.}, ``Extended reality over {5G}—standards evolution,'' \emph{IEEE J. Sel. Areas Commun.}, vol.~41, no.~6, pp. 1757--1771, 2023.

\bibitem{XR_1}
G.~Pan, S.~Xu, S.~Zhang, X.~Chen, and Y.~Sun, ``Quality of experience optimization for real-time {XR} video transmission with energy constraints,'' \emph{IEEE Trans. Veh. Technol.}, vol.~73, no.~10, pp. 15\,883--15\,888, 2024.

\bibitem{t3}
M.~Gapeyenko, V.~Petrov, S.~Paris, A.~Marcano, and K.~I. Pedersen, ``Standardization of extended reality {(XR)} over {5G} and {5G}-advanced {3GPP} new radio,'' \emph{IEEE Network.}, vol.~37, no.~4, pp. 22--28, 2023.

\bibitem{XR_2}
G.~Pan, S.~Xu, S.~Zhang, X.~Chen, and Y.~Sun, ``Quality of experience oriented cross-layer optimization for real-time {XR} video transmission,'' \emph{IEEE Trans. Circuits Syst. Video Technol.}, vol.~34, no.~8, pp. 7742--7755, 2024.

\bibitem{gesture1}
J.~Shen, J.~Dudley, G.~Mo, and P.~O. Kristensson, ``Gesture spotter: A rapid prototyping tool for key gesture spotting in virtual and augmented reality applications,'' \emph{IEEE Trans. Visual. Comput. Graphics.}, vol.~28, no.~11, pp. 3618--3628, 2022.

\bibitem{speech1}
Z.~Cai, Y.~Ma, and F.~Lu, ``Robust dual-modal speech keyword spotting for {XR} headsets,'' \emph{IEEE Trans. Visual. Comput. Graphics.}, vol.~30, no.~5, pp. 2507--2516, 2024.

\bibitem{ObjectTrackin}
Y.~Wu, H.~Sheng, Y.~Zhang, S.~Wang, Z.~Xiong, and W.~Ke, ``Hybrid motion model for multiple object tracking in mobile devices,'' \emph{IEEE Internet Things J.}, vol.~10, no.~6, pp. 4735--4748, 2023.

\bibitem{intro1}
K.~B. Letaief, Y.~Shi, J.~Lu, and J.~Lu, ``Edge artificial intelligence for {6G}: Vision, enabling technologies, and applications,'' \emph{IEEE J. Sel. Areas Commun.}, vol.~40, no.~1, pp. 5--36, 2022.

\bibitem{14}
E.~Li, L.~Zeng, Z.~Zhou, and X.~Chen, ``Edge {AI}: On-demand accelerating deep neural network inference via edge computing,'' \emph{IEEE Trans. Wireless Commun.}, vol.~19, no.~1, pp. 447--457, 2020.

\bibitem{t7}
F.~E. Subhan, A.~Yaqoob, C.~H. Muntean, and G.-M. Muntean, ``A survey on artificial intelligence techniques for improved rich media content delivery in a {5G} and beyond network slicing context,'' \emph{IEEE Commun. Surveys Tuts.}, pp. 1--1, 2024.

\bibitem{t6}
T.~Taleb \emph{et~al.}, ``Toward supporting {XR} services: Architecture and enablers,'' \emph{IEEE Internet Things J.}, vol.~10, no.~4, pp. 3567--3586, 2023.

\bibitem{intro3}
Y.~Shi, K.~Yang, T.~Jiang, J.~Zhang, and K.~B. Letaief, ``Communication-efficient edge {AI}: Algorithms and systems,'' \emph{IEEE Commun. Surveys Tuts.}, vol.~22, no.~4, pp. 2167--2191, 2020.

\bibitem{03}
F.~Wang, S.~Cai, and V.~K.~N. Lau, ``Decentralized {DNN} task partitioning and offloading control in {MEC} systems with energy harvesting devices,'' \emph{IEEE J. Sel. Top. Sign. Proces.}, vol.~17, no.~1, pp. 173--188, 2023.

\bibitem{20}
X.~Ye, Y.~Sun, D.~Wen, G.~Pan, and S.~Zhang, ``End-to-end delay minimization based on joint optimization of {DNN} partitioning and resource allocation for cooperative edge inference,'' in \emph{IEEE Proc. VTC-Fall'23}, 2023, pp. 1--7.

\bibitem{15}
C.~Fang \emph{et~al.}, ``{DRL}-driven joint task offloading and resource allocation for energy-efficient content delivery in cloud-edge cooperation networks,'' \emph{IEEE Trans. Veh. Technol.}, vol.~72, no.~12, pp. 16\,195--16\,207, 2023.

\bibitem{19}
T.~Wang, Y.~Deng, Z.~Yang, Y.~Wang, and H.~Cai, ``Parameterized deep reinforcement learning with hybrid action space for edge task offloading,'' \emph{IEEE Internet Things J.}, vol.~11, no.~6, pp. 10\,754--10\,767, 2024.

\bibitem{01}
S.~Bi, L.~Huang, H.~Wang, and Y.-J.~A. Zhang, ``Lyapunov-guided deep reinforcement learning for stable online computation offloading in mobile-edge computing networks,'' \emph{IEEE Trans. Wireless Commun.}, vol.~20, no.~11, pp. 7519--7537, 2021.

\bibitem{07}
H.~Wu, J.~Chen, T.~N. Nguyen, and H.~Tang, ``Lyapunov-guided delay-aware energy efficient offloading in {IIoT-MEC} systems,'' \emph{IEEE Trans. Ind. Inform.}, vol.~19, no.~2, pp. 2117--2128, 2023.

\bibitem{33}
Y.~Chen, N.~Zhang, Y.~Zhang, X.~Chen, W.~Wu, and X.~Shen, ``Energy efficient dynamic offloading in mobile edge computing for internet of things,'' \emph{IEEE Trans. Cloud Comput.}, vol.~9, no.~3, pp. 1050--1060, 2021.

\bibitem{34}
Y.~Dai, K.~Zhang, S.~Maharjan, and Y.~Zhang, ``Deep reinforcement learning for stochastic computation offloading in digital twin networks,'' \emph{IEEE Trans. Ind. Inform.}, vol.~17, no.~7, pp. 4968--4977, 2021.

\bibitem{29}
Y.~Xu, X.~Li, and J.~Zhang, ``Device-to-device content delivery in cellular networks: Multicast or unicast,'' \emph{IEEE Trans. Veh. Technol.}, vol.~67, no.~5, pp. 4401--4414, 2018.

\bibitem{22}
A.~Xu \emph{et~al.}, ``{QDRL}: Queue-aware online {DRL} for computation offloading in industrial internet of things,'' \emph{IEEE Internet Things J.}, vol.~11, no.~5, pp. 7772--7786, 2024.

\bibitem{04}
W.~He, S.~Guo, S.~Guo, X.~Qiu, and F.~Qi, ``Joint {DNN} partition deployment and resource allocation for delay-sensitive deep learning inference in {IoT},'' \emph{IEEE Internet Things J.}, vol.~7, no.~10, pp. 9241--9254, 2020.

\bibitem{24}
Z.~Xu, D.~Yang, C.~Yin, J.~Tang, Y.~Wang, and G.~Xue, ``A co-scheduling framework for {DNN} models on mobile and edge devices with heterogeneous hardware,'' \emph{IEEE Trans. Mob. Comput.}, vol.~22, no.~3, pp. 1275--1288, 2023.

\bibitem{31}
L.~Ale, S.~A. King, N.~Zhang, A.~R. Sattar, and J.~Skandaraniyam, ``{D3PG}: Dirichlet {DDPG} for task partitioning and offloading with constrained hybrid action space in mobile-edge computing,'' \emph{IEEE Internet Things J.}, vol.~9, no.~19, pp. 19\,260--19\,272, 2022.

\bibitem{18}
H.~Jiang, X.~Dai, Z.~Xiao, and A.~Iyengar, ``Joint task offloading and resource allocation for energy-constrained mobile edge computing,'' \emph{IEEE Trans. Mob. Comput.}, vol.~22, no.~7, pp. 4000--4015, 2023.

\bibitem{23}
Y.~Su, W.~Fan, L.~Gao, L.~Qiao, Y.~Liu, and F.~Wu, ``Joint {DNN} partition and resource allocation optimization for energy-constrained hierarchical edge-cloud systems,'' \emph{IEEE Trans. Veh. Technol.}, vol.~72, no.~3, pp. 3930--3944, 2023.

\bibitem{25}
J.-A. Lim and Y.~Kim, ``Real-time {DNN} model partitioning for {QoE} enhancement in mobile vision applications,'' in \emph{Annu. IEEE Commun. Soc. Conf. Sens., Mesh Ad Hoc Commun. Netw.}, 2022, pp. 407--415.

\bibitem{11}
M.~Gao \emph{et~al.}, ``Task partitioning and offloading in {DNN}-task enabled mobile edge computing networks,'' \emph{IEEE Trans. Mob. Comput.}, vol.~22, no.~4, pp. 2435--2445, 2023.

\bibitem{32}
B.~Yang, X.~Cao, X.~Li, Q.~Zhang, and L.~Qian, ``Mobile-edge-computing-based hierarchical machine learning tasks distribution for {IIoT},'' \emph{IEEE Internet Things J.}, vol.~7, no.~3, pp. 2169--2180, 2020.

\bibitem{facial}
A.~Casas-Ortiz, J.~Echeverria, N.~Jimenez-Tellez, and O.~C. Santos, ``Exploring the impact of partial occlusion on emotion classification from facial expressions: A comparative study of {XR} headsets and face masks,'' \emph{IEEE Access.}, vol.~12, pp. 44\,613--44\,627, 2024.

\bibitem{facial2}
B.~David-John, K.~Butler, and E.~Jain, ``Privacy-preserving datasets of eye-tracking samples with applications in xr,'' \emph{IEEE Trans. Visual. Comput. Graphics.}, vol.~29, no.~5, pp. 2774--2784, 2023.

\bibitem{05}
G.~Pan, H.~Zhang, S.~Xu, S.~Zhang, and X.~Chen, ``Joint optimization of video-based {AI} inference tasks in {MEC}-assisted augmented reality systems,'' \emph{IEEE Trans. Cogn. Commun. Netw.}, vol.~9, no.~2, pp. 479--493, 2023.

\bibitem{cvx}
M.~Grant and S.~Boyd, ``{CVX}: {Matlab} software for disciplined convex programming, version 2.1,'' \url{https://cvxr.com/cvx}, Mar. 2014.

\bibitem{35}
H.~Li, K.~Xiong, Y.~Lu, W.~Chen, P.~Fan, and K.~B. Letaief, ``Collaborative task offloading and resource allocation in small-cell {MEC}: A multi-agent {PPO}-based scheme,'' \emph{IEEE Trans. on Mob. Comput.}, pp. 1--13, 2024.

\bibitem{30}
Y.~Fan, J.~Ge, S.~Zhang, J.~Wu, and B.~Luo, ``Decentralized scheduling for concurrent tasks in mobile edge computing via deep reinforcement learning,'' \emph{IEEE Trans. Mob. Comput.}, vol.~23, no.~4, pp. 2765--2779, 2024.

\bibitem{36}
J.~Ren, G.~Yu, Y.~Cai, and Y.~He, ``Latency optimization for resource allocation in mobile-edge computation offloading,'' \emph{IEEE Trans. Wireless Commun.}, vol.~17, no.~8, pp. 5506--5519, 2018.

\end{thebibliography}

\end{document}